%% file: Reconstruction_arxiv_v3.tex
\documentclass[journal, onecolumn, 11pt]{IEEEtran}

\usepackage{amsmath, amsthm, amssymb}
\usepackage{latexsym}
\usepackage{graphicx}
\usepackage{verbatim}
\usepackage{mathrsfs}
\usepackage{float}
\usepackage[lofdepth, lotdepth]{subfig}
\usepackage{array}
\usepackage{url}
\usepackage{algorithm}
\usepackage[noend]{algorithmic}
\usepackage{cite}
\usepackage[table]{xcolor}
\usepackage{enumerate}
\usepackage{eucal}
\usepackage{stmaryrd}
\usepackage{mathtools}
\usepackage{pgf}
\usepackage{array}

\usepackage{tikz}
\usetikzlibrary{arrows, arrows.meta,calc,positioning,shapes.arrows}
\usetikzlibrary{decorations.pathreplacing}
\usetikzlibrary{chains,fit,shapes}
\usetikzlibrary{mindmap,trees}

\usepackage[top=0.8in, bottom=0.8in, left=0.6in, right=0.6in]{geometry}
\usepackage{mathtools}
\DeclarePairedDelimiter\ceil{\lceil}{\rceil}
\DeclarePairedDelimiter\floor{\lfloor}{\rfloor}

\theoremstyle{plain}
\newtheorem{theorem}{Theorem}
\newtheorem{proposition}[theorem]{Proposition}

\newtheorem{corollary}[theorem]{Corollary}

\theoremstyle{definition}
\newtheorem{definition}[theorem]{Definition}
\newtheorem{example}[theorem]{Example}
\newtheorem{remark}[theorem]{Remark}

\newcommand{\B}{{\mathcal B}}
\newcommand{\C}{{\mathcal C}}

\newcommand{\X}{{\mathcal X}}

\newcommand{\Y}{{\mathcal Y}}

\newcommand{\R}{{\mathcal R}}

\DeclareMathAlphabet{\mathbfsl}{OT1}{ppl}{b}{it} 

\newcommand{\ba}{{\mathbfsl a}}
\newcommand{\bb}{{\mathbfsl b}}
\newcommand{\bc}{{\mathbfsl c}}
\newcommand{\barc}{\overline{\mathbfsl c}}

\newcommand{\bu}{{\mathbfsl u}}

\newcommand{\by}{{\mathbfsl y}}

\newcommand{\bx}{{\mathbfsl{x}}}
\newcommand{\bz}{{\mathbfsl{z}}}

\newcommand{\bbZ}{{\mathbb Z}}

\renewcommand{\ge}{\geqslant}
\renewcommand{\le}{\leqslant}

\newcommand{\dec}{\textsc{Dec}}

\newcommand{\inv}{{\rm Inv}}

\newcommand{\BS}{{\cal B}^{\rm S}}
\newcommand{\BSt}{{\cal B}^{{\rm S}(t)}}
\newcommand{\BDt}{{\cal B}^{{\rm D}(t)}}
\newcommand{\BI}{{\cal B}^{\rm I}}
\newcommand{\BD}{{\cal B}^{\rm D}}
\newcommand{\BSI}{{\cal B}^{\rm SI}}
\newcommand{\BSD}{{\cal B}^{\rm SD}}
\newcommand{\BID}{{\cal B}^{\rm ID}}
\newcommand{\Bedit}{{\cal B}^{\rm edit}}
\newcommand{\CD}{{\cal C}_{\rm D}}
\newcommand{\CSD}{{\cal C}_{\rm SD}}
\newcommand{\Cedit}{{\cal C}_{\rm edit}}
\newcommand{\Nsys}{{N}_{\rm sys}}

\newcommand{\dH}{{d}_{\rm H}}

\newcommand{\etal}{{\em et al.}}

\begin{document}

\pagestyle{empty}

\title{Coding for Sequence Reconstruction for Single Edits}
\author{\IEEEauthorblockN{
Kui Cai\IEEEauthorrefmark{1},
Han Mao Kiah\IEEEauthorrefmark{2}, 
Tuan Thanh Nguyen\IEEEauthorrefmark{1},
and Eitan Yaakobi\IEEEauthorrefmark{3}}\\[2mm]
\IEEEauthorblockA{
\small
\IEEEauthorrefmark{1}%
Singapore University of Technology and Design, Singapore 487372\\
\IEEEauthorrefmark{2}%
School of Physical and Mathematical Sciences, Nanyang Technological University, Singapore 637371\\
\IEEEauthorrefmark{3}%
Department of Computer Science, Technion --- Israel Institute of Technology, Haifa, 32000 Israel\\
Emails: cai\_kui@sutd.edu.sg, hmkiah@ntu.edu.sg, tuanthanh\_nguyen@sutd.edu.sg, yaakobi@cs.technion.ac.il}

\thanks{The material in this paper will be presented in part at the 
	IEEE International Symposium on Information Theory in June 2020. 
	See~\cite{Kiah.ISIT.2020}.}
}

\maketitle

\vspace{-5mm}


\hspace{-3.5mm}\begin{abstract}
The sequence reconstruction problem, introduced by Levenshtein in 2001, 
considers a communication scenario where the sender {transmits} a codeword from some codebook 
and the receiver obtains multiple noisy reads of the codeword.  
The common setup assumes the codebook to be the entire space and 
the problem is to determine the minimum number of distinct reads that is required to reconstruct the transmitted codeword. 

Motivated by modern storage devices, we study a variant of the problem 
where the number of noisy reads $N$ is fixed.
Specifically, we design {\em reconstruction codes} that reconstruct a codeword from $N$ distinct noisy reads.
We focus on channels that introduce single edit {error} (i.e. a single substitution, insertion, or deletion) and their variants,
and design reconstruction codes for all values of $N$.
In particular, for the case of a single edit, we show that as the number of noisy reads increases, 
the number of redundant symbols required can be gracefully reduced from 
$\log_q n+O(1)$ to $\log_q \log_q n+O(1)$, and then to $O(1)$,
where $n$ denotes the length of a codeword.
We also show that these reconstruction codes are asymptotically optimal.
Finally, via computer simulations, we demonstrate that in certain cases,
reconstruction codes can achieve similar performance as classical error-correcting codes with less redundant symbols.  
\end{abstract}

%


\section{Introduction}
\label{sec:intro}

As our data needs surge, new technologies emerge to store these huge datasets.
Interestingly, besides promising ultra-high storage density, 
certain emerging storage media rely on technologies that provide users with {\em multiple cheap, albeit noisy, reads}. 
In this paper, we leverage on these multiple reads to increase the information capacity of these next-generation devices, or equivalently, {\em reduce the number of redundant bits}. 
Before we formally state our problem, we list two storage scenarios where multiple cheap reads are available to the user.

\begin{enumerate}[(a)]
\item {\bf DNA-based data storage}. 
In these data systems \cite{Church.etal:2012, Goldman.etal:2013, Yazdi.etal:2015b, Yazdi2017.portable, Organick.2018}, digital information is stored in native or synthetic DNA strands and 
to read the information, a user typically employs a sequencing platform like the popular Illumina sequencer or more recently, a nanopore sequencer.
During both synthesis and sequencing process, many copies of the same DNA strand are generated. 
For example, in most sequencing platforms, a DNA strand undergoes {\em polymerase chain reaction} (PCR) and multiple copies of the same strand are created. 
The sequencer then reads all copies and provides multiple (possibly) erroneous reads to the user (see Figure~\ref{fig:nano}).
In nanopore sequencers, these reads are often inaccurate and high-complexity read-alignment and consensus algorithms are required to reconstruct the original DNA strand from these noisy reads. 

To reduce the read-alignment complexity and improve the read accuracy, 
one may employ various coding strategies to design DNA information strands.
Yazdi \etal{} \cite{Yazdi2017.portable} proposed a simple coding strategy and verified it experimentally. 
Later, Cheraghchi \etal{} \cite{Cheraghchi.2019} provided a marker-based coding strategy that has provable reconstruction guarantees.

\item {\bf Racetrack memories}. 
Based on spintronic technology,  a racetrack memory, also known as {\em domain wall memory}, is composed of cells, also called {\em domains}, which are positioned on a tape-like strip (see Figure~\ref{fig:racetrack})
and are separated by {\em domain walls} \cite{Parkin.2008, Zhang.2016}.
The magnetization of a domain is programmed to store a single bit value, which can be read by sensing its magnetization direction. 
The reading mechanism is operated by a read-only port, called a {\em head}, together with a reference domain. 
Since the head is fixed, a shift operation is required in order to read all the domains and 
this is accomplished by applying shift current which moves the domain walls in one direction. 
Multiple heads can also be used in order to significantly reduce the read access latency 
of the memory. 

When these heads read overlapping segments, we have multiple noisy reads.
Recently, Chee \etal{} \cite{chee2018coding} leveraged on these noisy reads to correct shift errors in racetrack memories. 
They designed an arrangement of heads and devised a corresponding coding strategy to correct such errors with a {\em constant number} of redundant bits.
\end{enumerate}

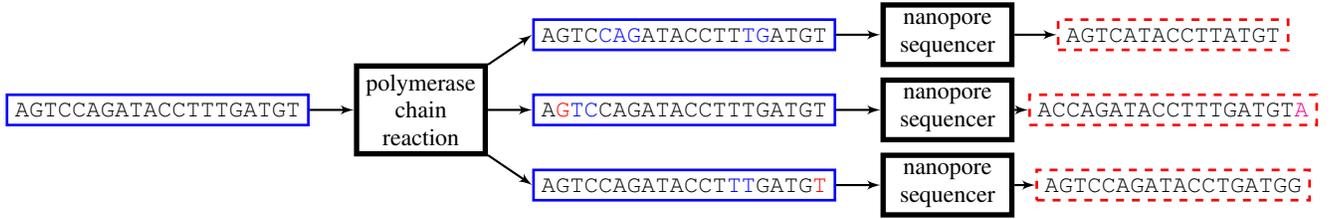
\begin{figure*}[!t]
	
	\begin{center}
		\footnotesize
		\begin{tikzpicture}
		
		\tikzset{read/.style = {rectangle, draw=red, dashed, line width=1pt, inner sep=3pt}}
		\tikzset{word/.style = {rectangle, draw=blue, line width=1pt, inner sep=3pt}}
		\tikzset{channel/.style = {rectangle, draw, line width=2pt, text width=15mm,align=center, inner sep = 3pt}}
		\tikzset{arrow/.style = {->,> = latex',black,thick}}
		\node[word](c) at (0,0)    {{\tt AGTCCAGATACCTTTGATGT}};
		\node[channel](pcr) at (3.5,0) {polymerase\\chain\\reaction};
		
		\node[word](c1) at (7,1)    {{\tt AGTC{\color{blue}CAG}ATACCTT{\color{blue}TG}ATGT}};
		\node[word](c2) at (7,0)    {{\tt A{\color{red}G}{\color{blue}TC}CAGATACCTTTGATGT}};
		\node[word](c3) at (7,-1)    {{\tt AGTCCAGATACCT{\color{blue}TT}GATG{\color{red}T}}};
		
		\node[channel](ch1) at (10.5,1) {nanopore\\sequencer};
		\node[channel](ch2) at (10.5,0) {nanopore\\sequencer};
		\node[channel](ch3) at (10.5,-1) {nanopore\\sequencer};
		
		\node[read](r1) at (13.5,1)  {{\tt AGTCATACCTTATGT}};
		\node[read](r2) at (13.5,0)  {{\tt ACCAGATACCTTTGATGT{\color{magenta}A}}};
		\node[read](r3) at (13.5,-1) {{\tt AGTCCAGATACCTGATGG}};

		\draw[arrow](c) to (pcr.west);
		\draw[arrow](pcr) to (c1.west);
		\draw[arrow](pcr) to (c2.west);
		\draw[arrow](pcr) to (c3.west);
		
		\draw[arrow](c1) to (ch1.west);
		\draw[arrow](c2) to (ch2.west);
		\draw[arrow](c3) to (ch3.west);
		
		\draw[arrow](ch1) to (r1.west);
		\draw[arrow](ch2) to (r2.west);
		\draw[arrow](ch3) to (r3.west);
		
		\end{tikzpicture}
		\vspace{2mm}
		\caption{\small Noisy reads from a nanopore sequencer. A DNA strand undergoes PCR and multiple copies of the same strand are created. The sequencer then reads all copies and provides multiple errononeous reads. Here, the basepairs coloured in {\color{blue}blue} were {\color{blue}deleted}, the basepairs coloured in {\color{magenta}magenta} were {\color{magenta}inserted}, while those coloured in {\color{red}red} have {\color{red}substitution errors}.
		\vspace{-2mm}
		}
		\label{fig:nano}
	\end{center}
\end{figure*}

\vspace{2mm}

Motivated by these applications, we study the following coding problem in a general setting.
Consider a data storage scenario where $N$ distinct noisy reads are provided. 
Our task is to design a {\em codebook} such that every codeword can be uniquely reconstructed from any $N$ distinct  noisy reads.
Hence, our fundamental problem is then: how large can this codebook be? 
Or equivalently, what is the {\em minimum number of redundancy {bits}}?

In this paper, we study in detail the case where the reads are affected by a single {\em edit} (a substitution, deletion, or insertion) and its variants.
In particular, for the case of a single edit, we show that as the number of noisy reads increases, 
the number of redundant bits required can be gracefully reduced from 
$\log_q n+O(1)$ to $\log_q \log_q n+O(1)$, and then to $O(1)$,
where $n$ denotes the length of a codeword.

\section{Problem Statement and Contributions}

Consider a data storage scenario described by an error-ball function.
Formally, given an input space $\X$ and {an} output space $\Y$, 
an {\em error-ball} function $B$ maps a {\em word} $\bx\in \X$ 
to a subset of {\em noisy reads} $B(\bx)\subset \Y$.
Given a code $\C\subseteq \X$, we define the {\em read coverage} of $\C$, 
denoted by $\nu(\C;B)$, to be the quantity
\begin{equation}\label{eq:read}
\nu(\C;B)\triangleq \max \Big\{|B(\bx)\cap B(\by)| : \, \bx,\by\in \C, \, \bx\ne\by \Big\}\, .
\end{equation}
\noindent 
In other words, $\nu(\C;B)$ is the maximum intersection between the error-balls of any two codewords in $\C$. 
The quantity $\nu(\C;B)$ was introduced by Levenshtein \cite{Levenshtein.2001}, 
where he showed that the number of reads%
\footnote{In the original paper, Levenshtein used the term ``channels'', instead of reads. 
	Here, we used the term ``reads'' to reflect the data storage scenario.}
required to reconstruct a codeword from $\C$ is at least $\nu(\C;B)+1$.
The problem to determine $\nu(\C;B)$ is referred to as the {\em sequence reconstruction problem}.

The sequence reconstruction problem was studied in a variety of storage and communication scenarios \cite{chee2018coding, Cheraghchi.2019, Gabrys.2018, Konstantinova:2008, Levenshtein.2009, Sala.2017, yehezkeally2018reconstruction}.
In these cases, $\C$ is usually assumed to be the entire input space, $\X$, (all words of some fixed length) or a classical error-correcting code.

However, in most storage scenarios, the number of noisy reads $N$ is a fixed system parameter and 
when $N$ is at most $\nu(\X;B)$, we are unable to {uniquely} reconstruct the codeword. 
This work looks at this regime where we design codes whose read coverage is strictly less than $N$.
Specifically, we say that $\C$ is an $(n,N;B)$-{\em reconstruction code} if 
$\C\subseteq \X$ and $\nu(\C;B)<N$.

This gives rise to a {\em new quantity of interest} that measures the 
{\em trade-off between codebook redundancy and read coverage}. 
Let $\Sigma_q$ denote the alphabet $\{0,1,\ldots, q-1\}$ with $q$ symbols.
Specifically, given $N$ and an error-ball function $B:\X \subseteq \Sigma_q^n \to \Y$, we study the quantity
	\begin{equation}\label{eq:code.red}
	\rho(n,N;B)\triangleq \min \Big\{n - \log_q|\C| : \C\subseteq \X,\, \nu(\C;B) < N \Big\}.
	\end{equation}
Note that the case $N=1$ is the classical model which has been studied for years in the design of error-correcting codes. 
Thus, we see the framework studied in this work as a natural extension of this classical model.

For a word $\bx\in \Sigma_q^n$, we consider the following error-ball functions.
Let $\BI_q(\bx)$, $\BD_q(\bx)$, and $\BS_q(\bx)$,
	denote the set of all words obtained from $\bx$ via one insertion, deletion, and at most one substitution, respectively.
	In this work, we study in detail the following error-balls:
\[
\BSD_q(\bx) \triangleq \BS_q(\bx)\cup \BD_q(\bx),~ 
\BSI_q(\bx) \triangleq \BS_q(\bx)\cup \BI_q(\bx),~
\BID_q(\bx) \triangleq \BI_q(\bx)\cup \BD_q(\bx), \mbox{ and }
\Bedit_q(\bx) \triangleq \BS_q(\bx)\cup \BI_q(\bx)\cup \BD_q(\bx).
\]
%
\vspace{-5mm}

\begin{example}\label{exa:toy}
	Consider the input space $\X=\Sigma_2^n$ with $q=2$.
	We consider the single-deletion error-ball $\BD_2$ and two different codebooks. 
	First, we look at the uncoded case, that is, when the codebook is the {input space $\Sigma_q^n$.
	Levenshtein} in his seminal work\cite{Levenshtein.2001} showed that $\nu\left(\Sigma_q^n;\BD_2\right)=2$.
	In other words, three distinct noisy versions of $\bx$ allow us to uniquely reconstruct $\bx$. 
	Hence, we have that $\rho(n,N;\BD_2)=0$ for $N\ge 3$.
	
	In contrast, to correct a single deletion, 
	we have the {classical} binary Varshamov-Tenengolts (VT) code $\C_V$ 
	whose redundancy is at most $\log_2(n+1)$~\cite{Levenshtein.1966}. 
	In this case, $\nu\left(\C_V;\BD_2 \right)=0$ and one noisy read is sufficient to recover a codeword.
	Furthermore, it can be shown that $\C_V$ is asymptotically optimal~\cite{Levenshtein.1966}, or,
	$\rho\left(n,1;\BD_2\right)=\log_2 n +\Theta(1)$ (see also Theorem~\ref{thm:ecc}).
	
	A natural question is then: how should we design the codebook when we have only two noisy reads? 
	Or, what is the value of $\rho\left(n,2;\BD_2\right)$?
	
	Recently, Chee \etal{} constructed a $\left(n,2;\BD_2\right)$-reconstruction code with $\log_2\log_2 n+O(1)$ redundant bits \cite{chee2018coding}.
	Hence, $\rho\left(n,2;\BD_2\right) \le  \log_2\log_2 n+O(1)$.
	In other words, even though {there are} only two noisy reads, {it is possible} to employ a coding strategy that encodes approximately {$\log n-\log\log n$} bits of information more than that of the VT code ${\rm VT}(n)$.
	We also show that this coding strategy is asymptotically optimal in Section~\ref{sec:exact}.
	In this paper, we extend this analysis and design such reconstruction codes for other error-balls.
\end{example}

\begin{figure}[!t]	
\footnotesize

\hspace{30mm}(a) Representation of bit-string $(0,0,1,1,0,1,0,1,1)$ using nine domains.	
	
  \begin{center}
  	
    \begin{tikzpicture}
		\tikzstyle{every path}=[very thick]
		
		\edef\sizetape{0.7cm}
		\tikzstyle{tmtape}=[draw,minimum size=\sizetape]
		\tikzstyle{tmhead}=[arrow box,draw,minimum size=.5cm,arrow box
		arrows={east:.25cm, west:0.25cm}]
		
		\tikzset{Arrow/.style = {line width=2mm, draw=gray, 
                -{Triangle[length=3mm,width=4mm]},
                shorten >=1mm, shorten <=1mm},}
		
		\tikzset{head/.style = {rectangle, draw, line width=2pt, text width=10mm,align=center, inner sep = 3pt}}
		
		\begin{scope}[start chain=1 going right,node distance=-0.15mm]
		    \node [on chain=1,tmtape,draw=none] (start){};
		    \node [on chain=1,tmtape,draw=none] {};
		    \node [on chain=1,tmtape] (input1){$\rightarrow$};
		    \node [on chain=1,tmtape] (input2){$\leftarrow$};
		    \node [on chain=1,tmtape] {$\leftarrow$};
		    \node [on chain=1,tmtape] (input3){$\rightarrow$};
		    \node [on chain=1,tmtape] {$\leftarrow$};
		    \node [on chain=1,tmtape] {$\rightarrow$};
		    \node [on chain=1,tmtape] {$\leftarrow$};
		    \node [on chain=1,tmtape] (tapend){$\leftarrow$};
		    \node [on chain=1,tmtape,draw=none] (end){};
		    \node [on chain=1] (end){};
		\end{scope}
		
		\draw[Arrow](input1)--(start);
		\draw[Arrow](end)--(tapend);
		
		\node[head](head1) at (1,1.5) {Head 1};
		\node[head](head2) at (2.5,1.5) {Head 2};
		\node[head](head3) at (4,1.5) {Head 3};
		\draw (7,-0.7) node {\textbf{Racetrack movement}};
		
		\path[->,draw] (head1.south) .. controls (0,1) and (1.2,1) ..  (input1.north);
		\path[->,draw] (head2.south) .. controls (2,1) and (1.9,1) ..  (input2.north);
		\path[->,draw] (head3.south) .. controls (3,1) and (3.5,1) ..  (input3.north);
	
	\end{tikzpicture}
	
	\vspace{1mm}
		
	\hspace{-7mm} (b) Readings of the three heads as the racetrack (or domains) shift.
	\newline
	
	\begin{tabular}{|c|ccc ccc ccc ccc|}
	\hline
	Shifts & $-2$ & $-1$ & $0$ & $1$ & $2$ & {\color{blue}$3$} & $4$ & $5$ & $6$ & $7$ & $8$ & $9$ \\ \hline
	Head 1 & -- & -- & -- & 0 & 0 &  {\color{blue}1} & 1 & 0 & 1& 0 & 1 & 1 \\
	Head 2 & -- & -- &  0 & 0 & 1 &  {\color{blue}1} & 0 & 1 & 0& 1 & 1 & -- \\
	Head 3 & 0 & 0 & 1 & 1 & 0 & {\color{blue}1} & 0 & 1 & 1& -- & -- & -- \\
	\hline
	\end{tabular}
	
	\vspace{5mm}

	\hspace{-11mm} (c) Corresponding noisy reads when {\color{blue}Shift 3} is not measured. \newline
	
	{$\!
	\begin{aligned}
	\text{Head 1 reading}: & (0,0,1,0,1,0,1,1),\\
	\text{Head 2 reading}: & (0,0,1,0,1,0,1,1),\\
	\text{Head 3 reading}: & (0,0,1,1,0,0,1,1).
	\end{aligned}$	
	}

    \caption{\footnotesize Noisy reads from a racetrack memory. 
    In a racetrack memory, the heads (reading mechanisms) are stationary, while the domains or cells are moved by a shift current applied in one direction. A stronger current may cause an {\em overshift}, resulting in a deleted bit in the readings.
    Here, we have an overshift at {\color{blue}Shift 3} and the corresponding {\color{blue}blue} bits are deleted from the output readings.\vspace{-7mm}
}
    \label{fig:racetrack}
  \end{center}
  
\end{figure}

\vspace{-5mm}

\subsection{Related Work}

We first review previous work {related to our problem} when there is only one noisy read, i.e. $N=1$.
In this case, we recover the usual notion of error-correcting codes.
Code constructions for the error-balls $B\in\{\BS_q,\BD_q,\BI_q\}$ have been studied extensively and we summarize the results here.


\pagebreak

\begin{theorem}[\hspace{-1mm}~\cite{Hamming:1950, Levenshtein.1966, MacWilliams.1977, Kulkarni.2013, Cai.2019}]\label{thm:ecc}
We have the following bounds for all $n$.
\begin{enumerate}[(i)]
\item $\log_q \big((q-1)n+1\big)  \le \rho\left(n,1;\BS_q\right)\le 
\big\lceil \log_q\big((q-1)n+1\big)\big\rceil$. 
\item We have that $\rho(n,1;\BD_q)=\rho(n,1;\BI_q)$. 
Also, $\rho(n,1;\BD_q)\ge \log_q \big((q-1)(n-1)\big)$, 
 $\rho(n,1;\BD_2)\le \log_2(n+1)$ and $\rho(n,1;\BD_q)\le 1 + \log_q n$ for $q\ge 3$; 
\item $\rho(n,1;\Bedit_2) \le 1+\log_2 n$ and $\rho(n,1;\Bedit_q) \le \log_q n + O(\log_q\log_q n)$.
\end{enumerate}
Therefore, we have that $\rho(n,1;B)=\log_q n +\Theta(1)$ for $B\in\{\BS_q,\BD_q,\BI_q\}$ and  $\rho(n,1;\Bedit_2) = \log_2 n+\Theta(1)$.
\end{theorem}

For the other error-ball functions of interest, the following is immediate from Theorem~\ref{thm:ecc}.

\begin{corollary}\label{cor:ecc}
	We have that $\rho(n,1;\BID_q)=\log_q n +\Theta(1)$. 
	For $B_q\in\{\BSD_q,\BSI_q\}$, we have that $\rho(n,1;B_2) = \log_2 n+\Theta(1)$ and
	$\log_q n - O(1) \le \rho(n,1;B_q) \le \log_q n + O(\log_q\log_q n)$.
\end{corollary}



%
%
%

When there {is} more than one noisy read, previous works usually focus on determining the maximum intersection size between two error-balls. 
Specifically, when the error-balls involve insertions only, deletions only and substitutions only, i.e. \smash{$B\in \{\BI_q, \BD_q, \BS_q\}$},  
the value of \smash{$\nu(\Sigma_q^n;B)$} was first determined by Levenshtein \cite{Levenshtein.2001}.
Later, Levenshtein's results were extended {in}~\cite{Gabrys.2018} for the case where the error-ball {involves} deletions only and 
$\C$ is a single-deletion error-correcting code.
Recently, {the authors of}~\cite{Sini.2019} investigated the case where errors are combinations of single substitution and single insertion. Furthermore, {they} also simplified the reconstruction algorithm when the number of noisy {copies exceeds}  the minimum required, i.e.  $N>\nu(\Sigma_q^n;B)+1$. 

Another recent variant of the Levenshtein sequence reconstruction problem was studied by 
the authors in~\cite{Junnila.2019}. Similar to our model, the authors consider the scenario where the number of reads is not sufficient to reconstruct a unique codeword. As with classical list-decoding, they determined the size of the {\em list of possible codewords}.

As mentioned above, the sequence reconstruction problem has been studied by Levenshtein and others for several error channels and distances. 
In many cases, such as for substitutions, the size of the set $B(\bx)\cap B(\by)$ does not depend on the specific choice of $\bx$ and $\by$, but only on their distance.
Specifically, let $\smash{\BSt_q}$ denote the radius-$t$ substitution error-ball and let $\dH(\bx,\by)$ denote the Hamming distance of $\bx$ and $\by$. 
The following is due to Levenshtein.

\begin{theorem}[{Levenshtein~\cite[Thm. 1, Lemma 3, Cor. 1]{Levenshtein.2001}}]\label{thm:tsub-coverage}
	Let $\bx$ and $\by$ be distinct words in $\Sigma_q^n$.
	If $\dH(\bx,\by)=d$, then
	\begin{equation}\label{eq:tsub}
	|\BSt_q(\bx)\cap \BSt_q(\by)| = N^S_n(q;t,d) \triangleq
	\sum_{i=0}^{t-\lceil\frac{d}{2}\rceil}\binom{n-d}{i}(q-1)^i
	\left(
	\sum_{k=d-t+i}^{t-i}\sum_{\ell=d-r+i}^{t-i}\binom{d}{k}\binom{d-k}\ell (q-2)^{d-k-\ell} \right)\, .
	\end{equation}
	
	\noindent Therefore, if $\C$ is a code with minimum Hamming distance $d$, then  $\nu\left(\C;\BSt_q\right)= N^S_n(q;t,d)$.
	Furthermore,
	\begin{equation}\label{eq:tsub-coverage}
	\nu\left(\Sigma_q^n;\BSt_q\right)= N^S_n(q;t,1) = q \sum_{i=0}^{t-1}\binom{n-i}{i}(q-1)^i.
	\end{equation}
\end{theorem}

Therefore, the quantity $\rho\left(n,N;\BSt_q\right)$ can be directly solved by finding the minimum Hamming distance of any valid code.
Denote by $A_q(n,d)$ the size of the largest length-$n$ $q$-ary code with minimum Hamming distance $d$. The following theorem is immediate from \eqref{eq:tsub-coverage}.

\begin{theorem}\label{thm:sub-N2}
For all $N\geq 1$, it holds that
\[\rho\left(n,N;\BSt_q\right) = n-\log_q A_q(n,d),\]
where $d$ is smallest integer such that $N^S_n(q,t,d)<N$. 
\end{theorem}
When $q=2$, it holds that $N^S_n(2;t,d) = N^S_n(2;t,d-1)$ for even values of $d${\cite{YB19}}, 
and therefore it is enough to consider only even values of $d$. 
Furthermore, for $d= 2t$, it holds that $N^S_n(2;t,d=2t) = \binom{2t}{t}$~\cite[Corollary 1]{Levenshtein.2001}, which implies that for all $t\geq 1$, 
$$\rho\left(n,N = \binom{2t}{t} +1 ;\BSt_2\right) = n-\log_2 A_2(n,2t),$$ 
and for all $1\leq N \leq \binom{2t}{t}$, $\rho\left(n,N ;\BSt_2\right) \ge n-\log_2 A_2(n,2t)$.

Next, for general $q$, we consider the case $t=1$, or when the error-ball is 
$\B^{{\rm S}(1)}_q=\BS_q$. 
From \eqref{eq:tsub-coverage}, we have the following
\begin{equation}\label{eq:char-sub}
N^S_n(q,1,d) = 
\begin{cases}
	q, & \text{if } d=1,\\
	2, & \text{if } d=2,\\
	0, & \text{if } d\ge 3.\\
\end{cases}
\end{equation}

Since $\log_q A_q(n,2)=n-1$, applying Theorem~\ref{thm:sub-N2} and combining with the results in Theorem~\ref{thm:ecc},
we have the following.

\begin{theorem}\label{thm:bs}
	Consider the error ball $\BS_q$. We have that
	\[
	\rho\left(n,N;\BS_q\right) = 
	\begin{cases}
	\log_q [(q-1)n+1]+\Theta(1), & \text{if } 1\le N\le 2,\\
	1, & \text{if } 3\le N\le q,\\
	0, & \text{if } N\ge q+1.
	\end{cases}
	\]
\end{theorem}

However, for the case of deletions, analogous results do not hold. 
Specifically, let $\BDt_q$ denote the radius-$t$ deletion-ball and 
we define the {\em Levenshtein distance} of two $q$-ary words $\bx$ and $\by$ to be smallest value of $t$ such that $\BDt_q(\bx)\cap\BDt_q(\by)$ is nonempty. In other words, $d_{\rm L}(\bx, \by)\triangleq\min\{t : \BDt_q(\bx)\cap\BDt_q(\by)\ne \varnothing \}$.
Then it can be shown that the Levenshtein distance is a metric.

Furthermore, it follows from {the} definition that  $|\BDt_q(\bx)\cap\BDt_q(\by)|=0$ whenever $t<d_{\rm L}(\bx,\by)$.
Unfortunately, when $t\ge d_{\rm L}(\bx,\by)$, the {Levenshtein distance} is unable to characterize the intersection size of the deletion-balls. An example is as follows. Consider $t=1$ and the following two pairs of words.
\[
\bx  = 0101,\quad \by = 1010, \mbox{ and } \bx'= 0111, \quad \by'=1110.
\]	
Then $d_{\rm L}(\bx,\by)=d_{\rm L}(\bx',\by')=1$.
Even though the two pairs have the same {Levenshtein distance}, the intersection sizes of their single-deletion-balls are different.
It is straightforward to verify that 
\[ |\BD_q(\bx)\cap \BD_q(\by)|=2, \mbox{ while } |\BD_q(\bx')\cap \BD_q(\by')|=1.\]

Hence, we require another property on pairs of words to characterize the intersection sizes.
Furthermore, unlike the case for substitutions, this example illustrates that designing reconstruction codes for the deletions channel is fundamentally different from designing codes in the {Levenshtein metric}.

\begin{table*}[!t]
	\centering
	\footnotesize
	\renewcommand{\arraystretch}{1.5}
	\begin{tabular}{lp{6cm}lp{2.8cm}p{5cm}} 
		\hline
		Code & Reconstruction Capabilities $(N;B)$ & Redundancy 
		& Lower Bound for\newline Optimal~Redundancy & Remark \\ \hline
		
		$\CD$ & $(2,\BI_q)$, $(2,\BI_q)$, $(3,\BID_q)$, $(3,\BID_q)$ &
		$\log_q\log_q n+O(1)$ & $\log_q\log_q n-O(1)$ &
		Theorem~\ref{thm:del}, Corollary~\ref{cor:id} and \cite{chee2018coding} \\ 
		
		$\CSD$ & $(3,\BSD_q)$, $(2,\BSI_q)$ &
		$\log_q\log_q n+O(1)$ & $\log_q\log_q n-O(1)$ &
		Theorem~\ref{thm:sd} \\ 
		
		$\Cedit$ & $(3,\Bedit_q)$, $(4,\Bedit_q)$ &
		$\log_q\log_q n+O(1)$ & $\log_q\log_q n-O(1)$ &
		Corollary~\ref{cor:edit} \\ 
		
		$\C_2$ & 	
		$(4,\BSD_q)$ for $q\ge 3$; $(4,\BSI_q)$ for $q\ge 2$; \newline 
		$(5,\Bedit_2)$; 
		$(N;\Bedit_q)$ for $q\ge 3$, $N\in\{5,6\}$ &
		$2$ & $1$ &
		Theorem~\ref{thm:even} \\ 
		
		$\C_1$ & $(n,4;\BSD_2)$, $(n,6;\Bedit_2)$ & 
		$1$ & $1-o(1)$ &
		Theorem~\ref{thm:even} \\ 
		
		$\C_0$ & $(N, \BSD_q)$ for $q\ge 4$ and $5\le N\le q+1$; \newline 
		$(N, \BSI_q)$ for $q\ge 3$ and $5\le N\le q+2$; \newline
		$(N,\Bedit_q)$for $q\ge 4$ and $7\le N\le q+3$. &
		$1$ & $1$ &
		Proposition~\ref{prop:spc},\newline  single-parity-check code\\ \hline 	
	\end{tabular}
	\caption{Reconstruction codes designed in this paper.}

	\label{code:summary}
\end{table*}

\begin{table*}[!t]
	\centering
	\scriptsize
	\renewcommand{\arraystretch}{1.2}
	\begin{tabular}{r| c  c c c c c}
		$N~\backslash~B$ & $\BS_2$ & $\BI_2$, $\BD_2$ & $\BID_2$ & $\BSI_2$ & $\BSD_2$ & $\Bedit_2$\\ \hline
		
		1 & 
		$\log_2 n +\Theta(1)$ & 
		$\log_2 n +\Theta(1)$ & 
		$\log_2 n +\Theta(1)$ & 
		$\log_2 n+ \Theta(1)$ & 
		$\log_2 n+ \Theta(1)$ & 
		$\log_2 n+ \Theta(1)$ \\
		
		2 & 
		$\log_2 n+ \Theta(1)$ & 
		$\log_2 \log_2 n + \Theta(1)$ $^\star$ & 
		$\log_2 n +\Theta(1)$ & 
		$\log_2 n +\Theta(1)$ &
		$\log_2 n +\Theta(1)$ &
		$\log_2 n +\Theta(1)$\\
		
		3 & 
		0 &
		0 & 
		$\log_2 \log_2 n + \Theta(1)$ $^\star$ &
		$\log_2 \log_2 n + \Theta(1)$ $^\star$ &
		$\log_2 \log_2 n + \Theta(1)$ $^\star$ &
		$\log_2 \log_2 n + \Theta(1)$ $^\star$ \\
		
		4 & 
		0 &
		0 & 
		$\log_2 \log_2 n + \Theta(1)$ $^\star$ &
		2 $^{\star\dagger}$ &
		1 $^{\star\ddagger}$ &
		$\log_2 \log_2 n + \Theta(1)$ $^\star$ \\
		
		5 & 
		0 &
		0 & 
		0 $^\star$ &
		0 $^\star$ &
		0 $^\star$ &
		2 $^{\star\dagger}$\\
		
		6 & 
		0 &
		0 & 
		0 $^\star$ &
		0 $^\star$ &
		0 $^\star$ &
		1 $^{\star\ddagger}$ \\
		
		$\ge 7$ & 
		0 &
		0 & 
		0 $^\star$ &
		0 $^\star$ &
		0 $^\star$ &
		0 $^\star$ \\
	\end{tabular}
	\caption{\footnotesize Asymptotically {\em exact} estimates on $\rho(n,N;B)$ for various error-ball functions when $q=2$. 
		Marked with $\star$ are entries derived in this work.
		Marked with $\dagger$ are entries where $\log_2 3\le \lim_{n\to \infty}\rho(n,N;B)\le 2$.
		Marked with $\ddagger$ are entries where $\lim_{n\to \infty}\rho(n,N;B)=1$.
	}
	\label{redundancy-q2}
\end{table*}

\begin{table*}[!t]
	\centering
	\scriptsize
	 \setlength{\tabcolsep}{4pt}
	\renewcommand{\arraystretch}{1.2}
	\begin{tabular}{r| c  c c c c c}
		$N \backslash B$ & $\BS_4$ & $\BI_4$, $\BD_4$ & $\BID_4$ & $\BSI_4$ & $\BSD_4$ & $\Bedit_4$\\ \hline
		
		1 & 
		$\log_4 n +\Theta(1)$ & 
		$\log_4 n +\Theta(1)$ & 
		$\log_4 n +\Theta(1)$ & 
		{\color{red}$\log_4 n+ O(\log_4\log_4 n)$} & 
		{\color{red}$\log_4 n+ O(\log_4\log_4 n)$} &
		{\color{red}$\log_4 n+ O(\log_4\log_4 n)$} \\
		
		2 & 
		$\log_4 n+ \Theta(1)$ & 
		$\log_4 \log_4 n + \Theta(1)$ $^\star$ & 
		$\log_4 n +\Theta(1)$ & 
		{\color{red}$\log_4 n+ O(\log_4\log_4 n)$} & 
		{\color{red}$\log_4 n+ O(\log_4\log_4 n)$} &
		{\color{red}$\log_4 n+ O(\log_4\log_4 n)$} \\
		
		3 & 
		1 ${^\ddagger}$ &
		0 & 
		$\log_4 \log_4 n + \Theta(1)$ $^\star$ &
		$\log_4 \log_4 n + \Theta(1)$ $^\star$ &
		$\log_4 \log_4 n + \Theta(1)$ $^\star$ &
		$\log_4 \log_4 n + \Theta(1)$ $^\star$ \\
		
		4 & 
		1 ${^\ddagger}$ &
		0 & 
		$\log_4 \log_4 n + \Theta(1)$ $^\star$ &
		2 $^{\star\dagger}$ &
		2 $^{\star\dagger}$ &
		$\log_4 \log_4 n + \Theta(1)$ $^\star$ \\
		
		5 & 
		0 &
		0 &
		0 &
		1 $^{\star\ddagger}$ &
		1 $^{\star\ddagger}$ &
		2 $^{\star\dagger}$\\
		
		6 & 
		0 &
		0 & 
		0 $^\star$ &
		1 $^{\star\ddagger}$ &
		0 $^\star$ &
		2 $^{\star\dagger}$ \\
		
		7 & 
		0 &
		0 & 
		0 $^\star$ &
		0 $^\star$ &
		0 $^{\star}$ &
		1 $^{\star\ddagger}$ \\

		$\ge 8$ & 
		0 &
		0 & 
		0 $^\star$ &
		0 $^\star$ &
		0 $^\star$ &
		0 $^\star$ \\
	\end{tabular}
	\caption{\footnotesize Asymptotically estimates on $\rho(n,N;B)$ for various error-ball functions when $q=4$.
		All entries {\em except} those marked in {\color{red}red} are asymptotically {\em exact}. 
		Marked with $\star$ are entries derived in this work.
		Marked with $\dagger$ are entries where $1\le \lim_{n\to \infty}\rho(n,N;B)\le 2$.
		Marked with $\ddagger$ are entries where $\rho(n,N;B)=1$.
	}
	\label{redundancy-q4}
\end{table*}

\subsection{Main Contributions}

In this work, we focus on the case where $2\le N \le \nu(\Sigma_q^n;B)$ with $B\in\{ \BS,\BI,\BD,\BID,\BSI,\BSD,\Bedit\}$.
In other words, we have more than one noisy reads, but not sufficient reads to uniquely reconstruct a $q$-ary word.
When $N=2$ and $B=\BD$, we have a recent code construction from \cite{chee2018coding} (see Example~\ref{exa:toy} and Theorem~\ref{thm:del}).
Specifically, in Section~\ref{sec:construction}, we make suitable modifications to this code construction and 
{show} that the resulting codes are $(n,N;B)$-reconstruction code for the error-balls of interest.
Table~\ref{code:summary} summarizes all reconstruction codes provided in this paper.

To do so, in Section~\ref{sec:intersection}, we study in detail the intersection of certain error-balls and 
derive the necessary and sufficient conditions for the size of an intersection.
Using these {characterizations}, we not only design the desired reconstruction codes, 
{but} we also obtain asymptotically tight lower bounds in Section~\ref{sec:exact} 
and in our companion paper~\cite{Chrisnata.arxiv.2020}.
We summarize a part of our results for the case $q\in\{2,4\}$ in Tables~\ref{redundancy-q2} and~\ref{redundancy-q4}, respectively.
We remark that all values of $\rho(n,N;B)$ have been determined asymptotically, except when
 $N\in \{1,2\}$ and $B\in\{\BSD_q, \BSI_q, \Bedit_q\}$ and $q\ge 3$.

In Section~\ref{sec:simulations}, we propose a simple bounded distance decoder for these reconstruction codes. 
Via computer simulations, we demonstrate that in certain cases,
reconstruction codes can achieve similar performance as classical error-correcting codes with less redundant symbols. 

\section{Combinatorial Characterization of the Intersection of Error-Balls}
\label{sec:intersection}

In the previous section, Theorem~\ref{thm:tsub-coverage} provides the read coverage $\nu(\Sigma_q,\BS)$ for the case of substitutions, and in this section, we determine the read coverage of the input space for our other error-balls of interest.
Namely, we determine $\nu(\Sigma_q^n;B)$ for $B\in\{\BD_q,\BI_q, \BSI_q, \BSD_q, \BID_q, \Bedit_q\}$.
In addition to computing the read coverage or maximum intersection size, 
we also characterize when the error-balls of a pair of words have intersection of a certain size.
This combinatorial characterization will be crucial in the code construction in Section~\ref{sec:construction}.

%


Besides the case of substitutions, Levenshtein also determined the read coverage for $B\in \{\BD,\BI\}$.

\begin{theorem}[{Levenshtein \cite[Eq. (26), (34)]{Levenshtein.2001}}]\label{thm:deletion-coverage}
Let $B\in\{\BI_q,\BD_q\}$.
If $\bx$ and $\by$ are distinct words in $\Sigma_q^n$, then
$|B(\bx)\cap B(\by)|\le 2$.
Therefore, we have that
$\nu(\Sigma_q^n;B)=2$ and $\rho(n,N;B)=0$ for $N\ge 3$. 
\end{theorem}

Furthermore, in an earlier work, Levenshtein showed that $\BD_q(\bx)\cap \BD_q(\by)=\varnothing$ 
if and only if $\BI_q(\bx)\cap \BI_q(\by)=\varnothing$ \cite{Levenshtein.1966}.
However, as pointed earlier, the Levenshtein distance is insufficient to characterization the intersection size.
To do so, we require the following notion of confusability. 

\pagebreak

\begin{definition}
Two words $\bx$ and $\by$ of length $n$ are said to be {\em Type-A-confusable}
if there exists subwords $\ba$, $\bb$, and $\bc$ such that the following holds.
\begin{enumerate}
\item[(A1)] $\bx=\ba\bc\bb$ and $\by=\ba\barc\bb$, where $\barc$ is the complement of $\bc$.
\item[(A2)] 
$\bc$ is one of the following forms: $(\alpha \beta)^m$, $(\alpha \beta)^m\beta$ for some $m\ge 1$ and $\alpha\ne\beta \in\Sigma_q$.
\end{enumerate}
\end{definition}

\begin{proposition}\label{char:deletion}
Let $B\in\{\BD_q,\BI_q\}$ and $\bx$ and $\by$ be distinct words in $\Sigma_q^n$.
\begin{itemize}
	\item If $\dH(\bx,\by)=1$, we have that $|\BD(\bx)\cap \BD(\by)|= 1$ while $|\BI(\bx)\cap \BI(\by)|= 2$. 
	\item If $\dH(\bx,\by)\ge 2$, we have that $|B(\bx)\cap B(\by)|= 2$ if and only if $\bx$ and $\by$ are Type-A-confusable.
\end{itemize} 
\end{proposition}

\begin{proof}
When the Hamming distance of $\bx$ and $\by$ is one, i.e. 
$\bx= \ba \alpha \bb$ and $\by = \ba \beta \bb$ where $\alpha \neq \beta$, 
we have  $\BD(\bx)\cap \BD(\by)=\{ \ba\bb\}$ and 
$\BI(\bx)\cap \BI(\by)=\{\ba\alpha \beta \bb, \ba \beta \alpha \bb\}$.	
	
Next, we consider the case that $\dH(\bx,\by)\ge 2$ 
and show that $|B(\bx)\cap B(\by)|= 2$ if and only if $\bx$ and $\by$ are Type-A-confusable.
We present the proof for the case where $B=\BD$ and the case for $B=\BI$ can be similarly proved.
Let $\bx=x_1x_2\cdots x_n$, $\by=y_1y_2\cdots y_n$ and $\BD(\bx)\cap \BD(\by)=\{\bz,\bz'\}$.
Since $\bx$ and $\by$ are of distance two, we set $i$ and $j$ be the smallest and largest indices, respectively,
where the two words differ.

We first consider $\bz \in\BD(\bx)\cap \BD(\by)$.
Let $\bz$ be obtained from $\bx$ by deleting index $k$ and from $\by$ by deleting index $\ell$.
We first claim that either $k\le i$ or $k\ge j$. Suppose otherwise that $i<k<j$ and we have two cases.
\begin{itemize}
\item When $\ell < j$, we consider the $(j-1)$th index of $\bz$. 
On one hand, since $k<j$ and we delete $x_k$ from $\bx$, the $(j-1)$th index of $\bz$ is $x_j$.
On the other hand, since we delete $y_\ell$ from $\by$, the $(j-1)$th index of $\bz$ is $y_j$.
Hence, $x_j=y_j$, yielding a contradiction.
\item When $\ell \ge j$, we consider the $i$th index of $\bz$. 
Proceeding as before, we conclude that $x_i=y_i$, which is not possible.
\end{itemize}

Without loss of generality, we assume that $k \le i$. A similar argument shows that $\ell\ge j$.
Therefore, we have that $y_{t}=x_{t+1}$ for $k \le t\le \ell-1$.
Recall that $x_t=y_t$ whenever $t\le i-1$ or $t\ge j+1$. 
Hence, we have that $x_k=x_{k+1}=\cdots = x_i$, $y_k=y_{k+1}=\cdots =y_{i-1}$, 
$x_{j+1}=x_{j+2}=\cdots = x_\ell$, $y_j=y_{j+1}=\cdots =y_{\ell}$.  
In summary, if we set $\ba=x_1x_2\cdots x_{i-1}=y_1y_2\cdots y_{i-1}$ and 
$\bb=x_{j+1}x_{j+2}\cdots x_n=y_{j+1}y_{j+2}\cdots y_{n}$,
then 
\[\bz=\ba x_{i+1}x_{i+2}\cdots x_j \bb=\ba y_{i}y_{i+1}\cdots y_{j-1} \bb.\]

Now, we consider $\bz'$, the other word in $\BD(\bx)\cap \BD(\by)$.
Since $\bz'$ is distinct from $\bz$ and  proceeding as before, we have that
\[\bz'=\ba x_{i}x_{i+1}\cdots x_{j-1} \bb=\ba y_{i+1}y_{i+2}\cdots y_{j} \bb.\]
Hence, $x_t=y_{t+1}=x_{t+2}$ for $i\le t\le j-2$. 
Since $\bz\ne \bz'$, we have that $x_i\ne x_{i+1}$ and 
$y_iy_{i+1}\cdots y_j=\overline{x_ix_{i+1}\cdots x_j}$.
Therefore, $\bx$ and $\by$ satisfy conditions (A1) and (A2) and are Type-A-confusable.

Conversely, suppose that $\bx$ and $\by$ are Type-A-confusable.
Hence, there exist subwords $\ba$, $\bb$, $\bc$ that satisfy conditions (A1) and (A2).
We further set $\bc_1$ and $\bc_2$ to be words obtained by deleting the first and last index from $\bc$, respectively.
Then $\ba\bc_1\bb$ and $\ba\bc_2\bb$ are two distinct subwords that belong to $\BD(\bx)\cap \BD(\by)$.
Since $|\BD(\bx)\cap \BD(\by)|\le 2$, we have that the intersection size must be exactly two.
\end{proof}

%

Proposition~\ref{char:deletion} can be extended to characterize the intersection sizes for error-balls involving either a single insertion or deletion.

%
%

\begin{proposition}[Single ID]\label{char:id}
Let $\bx$ and $\by$ be distinct words in $\Sigma_q^n$.
Then $|\BID_q(\bx)\cap\BID_q(\by)|\in\{0,2,3,4\}$.
Furthermore, we have $|\BID_q(\bx)\cap\BID_q(\by)|=4$  if and only if 
$\bx$ and $\by$ are Type-A-confusable.
Therefore, we have that  $\nu(\Sigma_q^n;\BID_q) = 4$ and 
$\rho(n,N;\BID_q)=0$ for $N\ge 5$.
\end{proposition}

\begin{proof}From \cite{Levenshtein.1966}, we have that for all $\bx$, $\by$,
\begin{equation}\label{eq:id}
\BD_q(\bx)\cap\BD_q(\by)= \varnothing \mbox{ if and only if } \BI_q(\bx)\cap \BI_q(\by) = \varnothing.
\end{equation}
Hence, since $\BID_q(\bx)\cap\BID_q(\by) = \left(\BD_q(\bx)\cap\BD_q(\by)\right)\cup \left(\BI_q(\bx)\cap\BI_q(\by)\right)$, it follows from Theorem~\ref{thm:deletion-coverage} that $|\BID_q(\bx)\cap\BID_q(\by)|\in\{0,2,3,4\}$.
Furthermore, when $|\BID_q(\bx)\cap\BID_q(\by)|=4$, we have that $|\BD_q(\bx)\cap\BD_q(\by)|=2$.
Proposition~\ref{char:deletion} then implies that $\bx$ and $\by$ are Type-A-confusable.
\end{proof}

When the error-balls include single substitutions, we require the following notion of confusability.

\begin{definition}
Two words $\bx$ and $\by$ of length $n$ are said to be {\em Type-B-confusable}
if there exists subwords $\ba$, $\bb$, $\bc$ and $\bc'$ such that the following hold:
\begin{enumerate}[(B1)]
\item $\bx=\ba\bc\bb$ and $\by=\ba\bc'\bb$;
\item $\{\bc,\bc'\}$ is of following form: $\{\alpha \beta ^m, \beta^m\alpha \}$ for some $m\ge 1$ and $\alpha\ne\beta \in\Sigma_q$.
\end{enumerate}
\end{definition}

\begin{proposition}[Single SD/SI]\label{char:sd}
Let $B\in\{\BSD_q,\BSI_q\}$ and $\bx$ and $\by$ be distinct words in $\Sigma_q^n$.
We have the following characterizations.

\begin{itemize}
\item If $\dH(\bx,\by)=1$, 
then $|\BSI_q(\bx)\cap \BSI(\by)|= q+2$ while $|\BSD(\bx)\cap \BSD(\by)|= q+1$.

\item If $\dH(\bx,\by)=2$, then $|B(\bx)\cap B(\by)|\in\{2,3,4\}$. Furthermore,
\begin{enumerate}[(i)]
\item $|B(\bx)\cap B(\by)|= 4$  if and only if 
$\bx$ and $\by$ are  Type-B-confusable with condition (B2) satisfied with $m=1$.
\item $|B(\bx)\cap B(\by)|= 3$  if and only if 
$\bx$ and $\by$ are  Type-B-confusable with condition (B2) satisfied with $m\ge 2$.
\end{enumerate}
\item If $\dH(\bx,\by)\ge 3$, then $|B(\bx)\cap B(\by)|\le 2$. 
Furthermore, $|B(\bx)\cap B(\by)|= 2$  if and only if $\bx$ and $\by$ are Type-A-confusable.
\end{itemize}

Therefore, $|\BSI_q(\bx)\cap \BSI_q(\by)|\le q+2$ and $|\BSD_q(\bx)\cap \BSD_q(\by)|\le \max\{q+1,4\}$.
Moreover, 
we have that $\nu(\Sigma_q^n;\BSI_q) = q+2$, $\nu(\Sigma_q^n;\BSD_q) = \max\{q+1,4\}$,  
$\rho(n,N;\BSI_q) =0$ for $N\ge q+3$ and $\rho(n,N;\BSD_q) =0$ for $N\ge \max\{q+2,5\}$.
\end{proposition}

\begin{proof}We consider the case $B=\BSD_q$ and the case for $\BSI_q$ can be similarly proved. 
Now, for all $\bx$, $\by$, 
\[
|\BSD_q(\bx)\cap\BSD_q(\by)|=|\BS_q(\bx)\cap \BS_q(\by)|+|\BD_q(\bx)\cap \BD_q(\by)|\,. 
\]

Hence, when $\dH(\bx,\by)=1$, it is immediate from \eqref{eq:tsub} and Proposition~\ref{char:deletion} that 
$|\BSD_q(\bx)\cap \BSD_q(\by)|=q+1$.

When $\dH(\bx,\by)=2$, it is again immediate from \eqref{eq:tsub} and Proposition~\ref{char:deletion} that 
$|\BSD_q(\bx)\cap \BSD_q(\by)|\in\{2,3,4\}$. 
Now, set $i$ and $j$ ($i<j$) to be the two indices where $\bx$ and $\by$ differ and we consider two cases.
\begin{itemize}
	\item $|\BSD_q(\bx)\cap\BSD_q(\by)|=3$ and so, $|\BD_q(\bx)\cap \BD_q(\by)|=1$.
	Set $\BD_q(\bx)\cap \BD_q(\by)=\{\bz\}$, 
	$\ba=x_1x_2\cdots x_{i-1}=y_1y_2\cdots y_{i-1}$, and 
	$\bb=x_{j+1}x_{j+2}\cdots x_n=y_{j+1}y_{j+2}\cdots y_{n}$.
	Then proceeding as the proof in Proposition~\ref{char:deletion}, we have that 
	\[\bz=\ba x_{i+1}x_{i+2}\cdots x_j \bb=\ba y_{i}y_{i+1}\cdots y_{j-1} \bb.\]
	Hence, we have $y_t=x_t=y_{t-1}$ for $i+1\le t\le j-1$ and 
	therefore, $\bx$ and $\by$ satisfy conditions (B1) and (B2) with $m\ge 2$.
	
	\item $|\BSD_q(\bx)\cap\BSD_q(\by)|=4$ and so, $|\BD_q(\bx)\cap \BD_q(\by)|=2$.
	Then Proposition~\ref{char:deletion} implies that $\bx$ and $\by$ are Type-A-confusable.
	Since $\bx$ and $\by$ differ at exactly two coordinates, 
	we have that they satisfy conditions (B1) and (B2) with $m=1$.
\end{itemize}

When $\dH(\bx,\by)\ge 3$, \eqref{eq:tsub} and Proposition~\ref{char:deletion} implies that 
$|\BSD_q(\bx)\cap \BSD_q(\by)|\le 2$. 
When $|\BSD_q(\bx)\cap\BSD_q(\bx)|=2$, it follows that $|\BD_q(\bx)\cap\BD_q(\by)|=2$.
Hence, Proposition~\ref{char:deletion} implies that $\bx$ and $\by$ are Type-A-confusable.
\end{proof}

Finally, we characterize the intersection sizes when the error-balls arise from single edits.

\begin{proposition}[Single Edit]\label{char:edit}
	Let $\bx$ and $\by$ be distinct words in $\Sigma_q^n$.
	We have the following characterizations.
	
	\begin{itemize}
		\item If $\dH(\bx,\by)=1$, 
		then $|\Bedit_q(\bx)\cap \Bedit_q(\by)|= q+3$.
		\item If $\dH(\bx,\by)=2$, then $|\Bedit_q(\bx)\cap \Bedit_q(\by)|\in \{2,4,6\}$. Furthermore,
		\begin{enumerate}[(i)]
			\item $|\Bedit_q(\bx)\cap \Bedit_q(\by)|= 6$  if and only if 
			$\bx$ and $\by$ are  Type-B-confusable with condition (B2) satisfied with $m=1$.
			\item $|\Bedit_q(\bx)\cap \Bedit_q(\by)|= 4$  if and only if 
			$\bx$ and $\by$ are  Type-B-confusable with condition (B2) satisfied with $m\ge 2$.
		\end{enumerate}
		\item If $\dH(\bx,\by)\ge 3$, then $|\Bedit(\bx)\cap\Bedit(\by)|\in\{0,1,2,4\}$. 
		Furthermore, $|\Bedit_q(\bx)\cap \Bedit_q(\by)|= 4$  if and only if $\bx$ and $\by$ are Type-A-confusable.
		
	\end{itemize}
	
	Therefore, $|\Bedit_q(\bx)\cap \Bedit_q(\by)|\le \max\{q+3,6\}$.
	Moreover,  we have that $\nu(\Sigma_q^n;\Bedit_q) = \max\{q+3,6\}$ and $\rho(n,N;\BSD_q) =0$ for $N\ge \max\{q+4,7\}$.
\end{proposition}

\begin{proof}
For all $\bx$, $\by$, we have that
\[ |\Bedit_q(\bx)\cap\Bedit_q(\by)| =|\BS_q(\bx)\cap \BS_q(\by)|+|\BD_q(\bx)\cap \BD_q(\by)|+|\BI_q(\bx)\cap \BI_q(\by)|. \]

As before, we consider the following cases according to the Hamming distance of $\bx$ and $\by$ and 
the proof proceeds in a similar manner as that of Proposition~\ref{char:sd}.

When $\dH(\bx,\by)=1$, we have that $|\Bedit_q(\bx)\cap \Bedit_q(\by)|=q+3$ from \eqref{eq:tsub} and Proposition~\ref{char:deletion}.

When $\dH(\bx,\by)=2$, we have that $|\BS_q(\bx)\cap \BS_q(\by)|=2$ from \eqref{eq:tsub}.
Proceeding as before, we have the following subcases.
\begin{itemize}
	\item $|\BD_q(\bx)\cap \BD_q(\by)|=|\BI_q(\bx)\cap \BI_q(\by)|=2$ if and only if 
	$\bx$ and $\by$ satisfy conditions (B1) and (B2) with $m= 1$.
	In other words,  $|\Bedit_q(\bx)\cap \Bedit_q(\by)|=6$. 
	
	\item $|\BD_q(\bx)\cap \BD_q(\by)|=|\BI_q(\bx)\cap \BI_q(\by)|=1$ if and only if 
	$\bx$ and $\by$ satisfy conditions (B1) and (B2) with $m\ge 2$.
	In other words,  $|\Bedit_q(\bx)\cap \Bedit_q(\by)|=4$.
	
	\item $|\BD_q(\bx)\cap \BD_q(\by)|=|\BI_q(\bx)\cap \BI_q(\by)|=0$ if and only if 
	$\bx$ and $\by$ are not Type-B-confusable. 
	In other words,  $|\Bedit_q(\bx)\cap \Bedit_q(\by)|=2$.
\end{itemize}

Finally, when $\dH(\bx,\by)\ge 3$, we have that $|\Bedit_q(\bx)\cap \Bedit_q(\by)|=|\BD_q(\bx)\cap \BD_q(\by)|+|\BD_q(\bx)\cap \BD_q(\by)|$. Hence, applying Proposition~\ref{char:deletion}, we have that $|\Bedit_q(\bx)\cap \Bedit_q(\by)|\le 4$.
Furthermore, if $|\Bedit_q(\bx)\cap\Bedit_q(\by)|\ge 3$, it follows that either $|\BD_q(\bx)\cap\BD_q(\by)|=2$ or $|\BI_q(\bx)\cap\BI_q(\by)|=2$. Proposition~\ref{char:deletion} then implies that $\bx$ and $\by$ are Type-A-confusable
and thus, $|\Bedit_q(\bx)\cap\Bedit_q(\by)|=4$.
\end{proof}

\section{Reconstruction Codes with $o(\log n)$ Redundancy}
\label{sec:construction}

Trivially, an $(n,N;B)$-reconstruction code is also an $(n,N';B)$-reconstruction code for $N'\ge N$.
Hence, it follows from Theorem~\ref{thm:ecc} and Corollary~\ref{cor:ecc}, 
that there exists an $(n,N;B)$-reconstruction code with $\log_q n+O(\log_q\log_q n)$ redundant symbols for all $B\in\left\{\BD_q,\BI_q,\BID_q,\BSD_q, \BSI_q, \Bedit_q\right\}$ and $N\ge 1$.
In this section, we provide reconstruction codes with redundancy $o(\log n)$ {when $N>1$}.

To do so, we recall a recent construction of an $(n, 2; \BD_2)$-reconstruction code provided by Chee \etal \cite{chee2018coding} 
in the context of racetrack memories. Crucial to this construction is the notion of period.

\begin{definition}
Let $\ell$ and $t$ be two positive integers where $\ell< t$. 
Then the word $\bu=u_1u_2\cdots u_t\in\Sigma_q^t$ is said to have {\em period $\ell$} or {\em $\ell$-periodic} 
if $u_i = u_{i+\ell}$ for all $1\le i \le t-\ell$.
We use $\R_q(n, \ell, t)$ to denote the set of all $q$-ary words $\bc$ of length $n$ 
such that the length of any $\ell'$-periodic ($\ell'\le \ell$) subword of $\bc$ is at most $t$.
\end{definition}

In \cite{chee2018coding}, it was shown that $\R_2(n,\ell,t)\ge 2^{n-1}$.
Following the proof, we extend the result for larger alphabet size.

\begin{proposition}[{Extended Result from \cite{chee2018coding}}]\label{prop:periodic-rll}
For $\ell\in\{1,2\}$, if $t\ge \ceil{\log n}+\ell$, we have that the size of $\R_q(n,\ell,t)$ is at least $q^{n-1}$.
\end{proposition}

Next, we define a $q$-ary syndrome similar to that of {the} VT-syndrome for binary alphabet.

\begin{definition}
	{Let $\bx=x_1x_2\cdots x_n\in\Sigma_q^n$. An} \emph{inversion} in $\bx$ is a pair of indices $(i,j)$ with $1\le i<j\le n$ and $x_i>x_j$.
	We use $\inv(\bx)$ to denote the total number of inversions in $\bx$. Formally,
	\begin{equation}
	\inv(\bx)\triangleq |\{(i,j): 1\le i<j\le n, \, x_i > x_j \}|.
	\end{equation}
\end{definition}

We are now ready to present our construction of an $(n, 2; \BD_q)$-reconstruction code.
Our construction extends the binary reconstruction code in \cite{chee2018coding} to the nonbinary case.
Here, we demonstrate its correctness for completeness and also because the key ideas are crucial to the constructions in Theorems~\ref{thm:sd} and~\ref{thm:even}.
{Furthermore, in the binary case, we improve the construction from \cite{chee2018coding} and reduce the redundancy by approximately one bit.}


\begin{theorem}[Single Deletion, $N=2$ \cite{chee2018coding}]\label{thm:del}
For $n , P >0$ with $P$ even, let $c\in \bbZ_{1+P/2}$ and $d\in \bbZ_2$.
Define $\CD(n;c,d)$ to be the set of all words $\bx=x_1x_2\cdots x_n \in \Sigma_q^n$ 
such that the following holds. 
\begin{enumerate}[(i)]
\item $\inv(\bx) = c \pmod{1+P/2}$.
\item $\sum_{i=1}^n x_i = d\pmod{q}$. 
\item $\bx$ belongs to $\R_q(n,2,P)$.
\end{enumerate}
Then $\CD(n;c,d)$ is an $(n, 2; \BD_q)$-reconstruction code.
Furthermore, if we set $P=\ceil{\log_q n}+2$, the code $\CD(n;c,d)$ has redundancy 
$2+\log_q(1+P/2)=\log_q \log_q n+O(1)$ 
 for some choice of $c$ and $d$.
\end{theorem}

\begin{proof}
We prove by contradiction.
Suppose that $\bx$ and $\by$ are two distinct words in $\CD(n;c,d)$ with $|\BD_q(\bx)\cap \BD_q(\by)|=2$.
%
Then Proposition~\ref{char:deletion} states that $\bx$ and $\by$ are Type-A-confusable.
In other words, there exist substrings $\ba$, $\bb$, $\bc$ such that $\bx=\ba\bc\bb$, $\by=\ba\barc\bb$
and $\bc$ has period exactly two.

Note that since the weights of $\bx$ and $\by$ have the same parity, 
we assume without loss of generality that  $\bc=(\alpha\beta)^m$ for some $m\ge 1$ and $\alpha>\beta$.
Then by construction, 
\begin{equation}\label{eq:del-proof}
\inv(\bx)-\inv(\by)= 0 \pmod{1+P/2}.
\end{equation}
On the other hand, since $\bx=\ba\bc\bb$ and $\by=\ba\barc\bb$,
the left-hand side of \eqref{eq:del-proof} evaluates to
\begin{align*}
\inv(\bx)-\inv(\by) &= \inv((\alpha\beta)^m) - \inv((\beta\alpha)^m)\\
& = [m+(m-1)+\cdots +1]-[(m-1)+\cdots +1] = m.
\end{align*}
However, since $\bc$ is a subword of $\bx$ with period two, we have that $2m\le P$,
and so, $m\ne 0\pmod{1+P/2}$, arriving at a contradiction.

Finally, we derive the upper bound on the redundancy.
When $P=\ceil{\log_q n}+2$, we have from Proposition~\ref{prop:periodic-rll} that the size of $\R_q(n,2,P)$ is at least $q^{n-1}$.
Since we have $q(1+P/2)$ choices for the syndromes $(c,d)$, by pigeonhole principle, there is a choice of $c$ and $d$ such that $\CD(n;c,d)$ has size at least $q^{n-2}/(1+P/2)$. The upper bound on the redundancy then follows.
\end{proof}

In a similar manner, we show that the code $\CD(n;c,d)$ is capable of reconstructing codewords from noisy reads affected by single insertions or deletions.

\begin{corollary}[Single Insertion/Deletion, $N\in\{3,4\}$]\label{cor:id}
Let $\CD(n;c,d)$ be as defined in Theorem~\ref{thm:del}. 
Then $\CD(n;c,d)$ is an $(n, N; \BID_q)$-reconstruction code for $N\in\{3,4\}$.
\end{corollary}

\begin{proof}
If two distinct words $\bx$ and $\by$ have $|\BID_q(\bx)\cap \BID_q(\by)|\ge N\ge 3$,
then $|\BD_q(\bx)\cap \BD_q(\by)|=2$ or $|\BI_q(\bx)\cap \BI_q(\by)|=2$. 
Suppose that  $|\BI_q(\bx)\cap \BI_q(\by)|=2$.
Since the Hamming distance of $\bx$ and $\by$ is at least two, 
Proposition~\ref{char:id} states that $\bx$ and $\by$ are Type-A-confusable, and 
hence, $|\BD_q(\bx)\cap \BD_q(\by)|=2$, contradicting Theorem~\ref{thm:del}. 
Thus, $\CD(n;c,d)$ is an $(n, N; \BID_q)$-reconstruction code for $N\in\{3,4\}$.
%
\end{proof}

When $B\in\{\BSD_q,\BSI_q\}$, we make suitable modifications to the code $\CD(n;c,d)$ to correct (possibly) a single substitution.

\begin{theorem}[Single Substitution/Deletion, $N=3$]\label{thm:sd}
For $n , P >0$, let $c\in \bbZ_{1+P}$ and $d\in \bbZ_2$.
Define $\CSD(n;c,d)$ to be the set of all words $\bx=x_1x_2\cdots x_n \in \Sigma_q^n$ 
such that the following holds. 
\begin{enumerate}[(i)]
\item $\inv(\bx) = c \pmod{1+P}$.
\item $\sum_{i=1}^n x_i = d\pmod{q}$. 
\item $\bx$ belongs to $\R_q(n,1,P)$.
\end{enumerate}
Then $\CSD(n;c,d)$ is an $(n, 3; B)$-reconstruction code for $B\in\{\BSD_q,\BSI_q\}$.
Furthermore, if we set $P=\ceil{\log_q n}+1$, the code $\CSD(n;c,d)$ has redundancy $2+\log_q(1+P)=\log_q \log_q n+O(1)$ 
 for some choice of $c$ and $d$.
\end{theorem}

\begin{proof}
We prove for the error-ball function $\BSD_q$ and prove by contradiction.
Suppose that $\bx$ and $\by$ are two distinct words in $\CSD(n;c,d)$ with $|\BSD_q(\bx)\cap \BSD_q(\by)|\ge 3$.
Since $\bx$ and $\by$ have the same parity, the Hamming distance of $\bx$ and $\by$ is at least two.
Then Proposition~\ref{char:sd} states that $\bx$ and $\by$ are Type-B-confusable.
Without loss of generality, let $\bx=\ba\alpha \beta^m\bb$, $\by=\ba\beta^m\alpha\bb$ with $\alpha>\beta$.

As before, we have 
\begin{equation}\label{eq:sd-proof}
\inv(\bx)-\inv(\by)= 0 \pmod{1+P}.
\end{equation}
Now, the left-hand side of \eqref{eq:sd-proof} evaluates to 
\[\inv(\bx)-\inv(\by) = \inv(\alpha\beta^m) - \inv(\beta^m\alpha)= m\]
However, since $\bx$ belongs to $\R_q(n,1,P)$, we have that $m\le P$, arriving at a contradiction.

As before, we derive the upper bound on the redundancy.
When $P=\ceil{\log n}+1$, we have from Proposition~\ref{prop:periodic-rll} that the size of $\R_q(n,1,P)$ is at least $q^{n-1}$.
Since we have $q(P+1)$ choices for the syndromes $(c,d)$, by pigeonhole principle, there is a choice of $c$ and $d$ such that $\CSD(n;c,d)$ has size at least $q^{n-2}/(P+1)$. The upper bound on the redundancy then follows.
\end{proof}

To correct a single edit with three or four reads, we make a small modification to $\CSD(n;c,d)$.

\begin{corollary}[Single Edit, $N\in\{3,4\}$]\label{cor:edit}
For $n , P >0$, let $c\in \bbZ_{1+P}$ and $d\in \bbZ_2$.
Define $\Cedit(n;c,d)$ to be the set of all words $\bx=x_1x_2\cdots x_n\in \Sigma_q^n$ 
such that the following holds. 
\begin{enumerate}[(i)]
\item $\inv(\bx) = c \pmod{1+P}$.
\item $\sum_{i=1}^n x_i = d\pmod{q}$. 
\item $\bx$ belongs to $\R_q(n,2,P)$.
\end{enumerate}Then $\Cedit(n;c,d)$ is an $(n, N; \Bedit_q)$-reconstruction code for $N\in\{3,4\}$.
Furthermore, if we set $P=\ceil{\log n}+2$, the code $\Cedit(n;c,d)$ has redundancy $2+\log_q(1+P)=\log_q\log_qn+O(1)$ 
for some choice of $c$ and $d$.
\end{corollary}

\begin{proof}
	Suppose that $\bx$ and $\by$ are two distinct words in $\Cedit(n;c,d)$ with $|\Bedit_q(\bx)\cap \Bedit_q(\by)|\ge N\ge 3$.
	Now, when $\dH(\bx,\by)\ge 2$, Proposition~\ref{char:edit} states that $|\Bedit_q(\bx)\cap \Bedit_q(\by)|\in\{0,1,2,4,6\}$.
	Hence, we necessarily have $|\Bedit_q(\bx)\cap \Bedit_q(\by)|\ge 4$.
	
Then we have two possibilities.
\begin{itemize}
\item $|\BSD_q(\bx)\cap \BSD_q(\by)|\ge 4$ or $|\BSI_q(\bx)\cap \BSI_q(\by)|\ge 4$. Note that since $\bx\in \R_q(n,2,P)$, we have $\bx\in \R_q(n,1,P)$.
Hence, following the proof of Theorem~\ref{thm:sd}, we obtain a contradiction.
\item $|\BD_q(\bx)\cap \BD_q(\by)|= |\BI_q(\bx)\cap \BI_q(\by)|=2$. Since $1+P/2\le 1+P$, we can follow the proof of Theorem~\ref{thm:del} to obtain a contradiction.
\end{itemize}

As before, we derive the upper bound on the redundancy.
When $P=\ceil{\log_q n}+2$, we have from Proposition~\ref{prop:periodic-rll} that the size of $\R_q(n,2,P)$ is at least $q^{n-1}$.
Since we have $q(P+1)$ choices for the syndromes $(c,d)$, by pigeonhole principle, there is a choice of $c$ and $d$ such that $\Cedit(n;c,d)$ has size at least $q^{n-2}/(P+1)$. The upper bound on the redundancy then follows.
\end{proof}

Our final code constructions are straightforward modifications of the usual single-parity-check codes.
For completeness, we define the single-parity-check and state their reconstruction capabilities.

Let the \emph{single-parity-check code} $\C_0$ be the following code
\[ \C_0 \triangleq \Bigg\{\bx=x_1x_2\cdots x_n\in \Sigma_q^n :\sum_{i=1}^{n} x_{i} = 0\pmod{q} \Bigg\}\, .\]
Then clearly, $\C_0$ has one redundant symbol.

\begin{proposition}\label{prop:spc}
	Let $\C_0$ be as defined above. 
	Then $\C_0$ is an $(n,N;\BSD_q)$-reconstruction code for $q\ge 4$ and $5\le N\le q+1$;
	an $(n,N;\BSI_q)$-reconstruction code for $q\ge 3$ and $5\le N\le q+2$; and
	an $(n,N;\Bedit_q)$-reconstruction code for $q\ge 4$ and $7\le N\le q+3$.
\end{proposition}

\begin{proof}
	In all cases, when $\bx$ and $\by$ are distinct words in $\C_0$, we have $\dH(\bx,\by)\ge 2$.
	Then applying Propositions~\ref{char:sd} and~\ref{char:edit}, we have that 
	$|\BSD_q(\bx)\cap \BSD_q(\by)|\le 4$, $|\BSI_q(\bx)\cap \BSI_q(\by)|\le 4$ and $|\Bedit_q(\bx)\cap \Bedit_q(\by)|\le 6$.
	The reconstruction properties then follow.
\end{proof}

Our next two code constructions introduce one and two symbols of redundancy, respectively.
In addition to taking the `parity' symbol of all coordinates, we take the `parity' of all {\em even} coordinates.
\begin{align*}
\C_1 & \triangleq \Bigg\{\bx=x_1x_2\cdots x_n\in \Sigma_q^n :\sum_{i=1}^{\floor{n/2}} x_{2i} = 0\pmod{q} \Bigg\}\, ,\\
\C_2 &\triangleq \Bigg\{\bx=x_1x_2\cdots x_n\in \Sigma_q^n :\sum_{i=1}^{\floor{n/2}} x_{2i} = 0\pmod{q} \text{ and } \sum_{i=1}^{n} x_{i} = 0\pmod{q} \Bigg\}\, .
\end{align*}
\noindent Clearly, $\C_1$ and $\C_2$ have one and two redundant symbols, respectively.

\pagebreak

\begin{theorem}\label{thm:even}
	Let $\C_1$ and $\C_2$ be as defined above.
	\begin{enumerate}[(i)]
		\item Then $\C_1$ is an $(n,4;\BSD_2)$-reconstruction code, and
		an $(n,6;\Bedit_2)$-reconstruction code.
		\item Then $\C_2$ is an $(n,4;\BSD_q)$-reconstruction code for $q\ge 3$; 
		an $(n,4;\BSI_q)$-reconstruction code for $q\ge 2$; 
		an $(n,5;\Bedit_2)$-reconstruction code; and 
		an $(n,N;\Bedit_q)$-reconstruction code for $q\ge 3$ and $N\in\{5,6\}$.
	\end{enumerate} 
\end{theorem}

\begin{proof}
	We prove for the error-ball $\BSD_q$ and the other error-balls follow similarly.
	
	When $q=2$, suppose that $\bx$ and $\by$ are two distinct words in $\C_1$ with $|\BSD_2(\bx)\cap \BSD_2(\by)|=4$.
	Proposition~\ref{char:sd} states that $\bx$ and $\by$ are Type-B-confusable with $m=1$.
	Then $\bx=\ba\alpha\beta\bb$, $\by=\ba\beta\alpha\bb$ with $\alpha\ne \beta$.
	Then $\sum_{i=1}^{\floor{n/2}} x_{2i}- y_{2i}= \beta-\alpha \ne 0 \pmod{q}$,  a contradiction.
	
	When $q\ge 3$, suppose that $\bx$ and $\by$ are two distinct words in $\C_2$ with $|\BSD_q(\bx)\cap \BSD_q(\by)|=4$. 
	Since $\dH(\bx,\by)\ge 2$, Proposition~\ref{char:sd} again states that $\bx$ and $\by$ are Type-B-confusable with $m=1$.
	Proceeding as before, we obtain a contradiction.
\end{proof}

The next section shows that the code constructions in this section are essentially optimal.

\section{Asymptotically Exact Bounds for the Redundancy of Reconstruction Codes}
\label{sec:exact}

In this section, we fix an error-ball $B\in\{\BI_q,\BD_q,\BID_q,\BSD_q, \BSI_q, \Bedit_q\}$ and a number of reads $N$ and 
determine the asymptotic value of the optimal redundancy $\rho(n, N; B)$. 
Formally, we show that $f(n)-\epsilon(n)\le \rho(n, N; B)\le f(n)+\epsilon(n)$ for some functions $f(n)$ and $\epsilon(n)$, 
where $\epsilon(n)=\Theta(1)$. 
In some cases, we also show that the absolute value of $\epsilon(n)$ is bounded by a constant at most one.
We remark that all values of $\rho(n,N;B)$ have been determined asymptotically, except when
$N\in \{1,2\}$ and $B\in\{\BSD_q, \BSI_q, \Bedit_q\}$ and $q\ge 3$.

Now, Section~\ref{sec:construction} {provides} certain code constructions and thus, we obtain upper bounds on the redundancy.
In what follows, we prove the corresponding lower bounds and hence, demonstrating the optimality of our code constructions.
To this aim, we make use of the following theorem 
that is demonstrated in our companion paper~\cite{Chrisnata.arxiv.2020}. 
The theorem is obtained from a careful analysis on the size of certain clique covers in related graphs and
interested readers may refer to the arXiv version for details.

\begin{theorem}\label{thm:confusable-optimal}
	Let $\C\subseteq \Sigma_q^n$.
	\begin{enumerate}[(i)]
		\item If every pair of distinct words are not Type-A-confusable, then $n- \log_q |\C|\ge \log_q\log_q n - O(1)$.
		\item If every pair of distinct words are not Type-B-confusable, then $n- \log_q |\C|\ge \log_q\log_q n - O(1)$.
		\item If every pair of distinct words are not Type-B-confusable with $m=1$, then $n- \log_q |\C|\ge 1-o(1)$.
		\item Suppose $q=2$. If every pair of distinct words have Hamming distance at least two and are not Type-B-confusable with $m=1$, then $n- \log_2 |\C|\ge \log_2 3-o(1)$.
	\end{enumerate}
\end{theorem}

First, we determine the optimal redundancy for the error-balls $\BD_q$ and $\BI_q$.

\begin{theorem}
	Let $B\in \{\BD_q, \BI_q \}$.
	Then 
	\[
	\rho(n,N;B) = 
	\begin{cases}
	\log_q n +\Theta(1), & \mbox{if } N=1,\\
	\log_q \log_q n +\Theta(1), & \mbox{if } N=2,\\
	0, & \mbox{if } N\ge 3.
	\end{cases}
	\]
\end{theorem}

\begin{proof}
	The case for $N=1$ follows from Theorem~\ref{thm:ecc} while the case for $N\ge 3$ follows Theorem~\ref{thm:deletion-coverage}.
	For $N=2$, the codes in Theorem~\ref{thm:del} have redundancy $\log_q\log_q n + O(1)$. 
	Therefore, it remains to provide the corresponding lower bound.
	
	If $\C$ is an $(n,2;B)$-reconstruction code, then we claim that every pair of distinct words $\bx$ and $\by$ are not Type-A-confusable. Suppose otherwise. Then Proposition~\ref{char:deletion} implies that $|B(\bx)\cap B(\by)|=2$, a contradiction. Hence, applying Theorem~\ref{thm:confusable-optimal}, we have that $\rho(n,2;B) \ge \log_q\log_q n - O(1)$.
\end{proof}

\begin{theorem}
	Consider the error-ball $\BID_q$. We have that
	\[
	\rho(n,N;\BID_q) = 
	\begin{cases}
	\log_q n +\Theta(1), & \mbox{if } N\in\{1,2\},\\
	\log_q \log_q n +\Theta(1), & \mbox{if } N\in\{3,4\},\\
	0, & \mbox{if } N\ge 5.
	\end{cases}
	\]
\end{theorem}

\begin{proof}
	The case for $N=1$ follows from Theorem~\ref{thm:ecc} while the case for $N\ge 5$ follows from Proposition~\ref{char:id}.
	
	For $N=2$, we first claim that $\C$ is an $(n,1;\BD_q)$-reconstruction code if and only if $\C$ is an $(n,2;\BID_q)$-reconstruction code. 
	This follows directly from \eqref{eq:id}. Hence, we have that $\rho(n,2;\BID_q)= \rho(n,1;\BID_q)=\log_q n +\Theta(1)$.
	
	When $N\in\{3,4\}$, the codes in Corollary~\ref{cor:id} have redundancy $\log_q\log_q n + O(1)$ and it remains to provide the corresponding lower bound.
	If $\C$ is an $(n,N;\BID_q)$-reconstruction code, then we claim that every pair of distinct words $\bx$ and $\by$ are not Type-A-confusable. Suppose otherwise. Then Proposition~\ref{char:deletion} implies that $|\BID_q(\bx)\cap \BID_q(\by)|=4\ge N$, a contradiction. Again, applying Theorem~\ref{thm:confusable-optimal}, we have that $\rho(n,2;\BID_q) \ge \log_q\log_q n - O(1)$.
\end{proof}

\begin{theorem}
	Consider the error-ball $\BSD_q$. We have the following.
	\begin{itemize}
		\item If $q=2$, we have that
		\[
		\rho(n,N;\BSD_2) = 
		\begin{cases}
		\log_2 n +\Theta(1), & \mbox{if } N\in\{1,2\},\\
		\log_2 \log_2 n +\Theta(1), & \mbox{if } N=3,\\
		1-o(1), & \mbox{if } N=4,\\
		0, & \mbox{if } N\ge 5.
		\end{cases}
		\]
		\item If $q\ge 3$, we have that
		
		\[
		\rho(n,N;\BSD_q) = 
		\begin{cases}
		\log_q n +O(\log_q\log_q n), & \mbox{if } N\in\{1,2\},\\
		\log_q \log_q n +\Theta(1), & \mbox{if } N=3,\\
		2 - \epsilon(n)  & \mbox{if } N = 4\\
		1, & \mbox{if } 5\le N\le q+1,\\
		0, & \mbox{if } N\ge q+2.
		\end{cases}
		\]
		Here, $0\le \epsilon(n)\le 1$.
	\end{itemize}
	
\end{theorem}

\begin{proof}
	The case for $N\in\{1,2\}$ follows from Corollary~\ref{cor:ecc} while the case for $N\ge \max\{5,q+2\}$ follows from Proposition~\ref{char:sd}.
	
%
	When $N=3$, the codes in Theorem~\ref{thm:sd} have redundancy $\log_q\log_q n + O(1)$ and it remains to provide the corresponding lower bound.
	If $\C$ is an $(n,3;\BSD_q)$-reconstruction code, then we claim that every pair of distinct words $\bx$ and $\by$ are not Type-B-confusable. Suppose otherwise. Then Proposition~\ref{char:sd} implies that $|\BSD_q(\bx)\cap \BSD_q(\by)|\ge 3$, a contradiction. Again, applying Theorem~\ref{thm:confusable-optimal}, we have that $\rho(n,2;\BSD_q) \ge \log_q\log_q n - O(1)$.
	
	When $N=4$, the codes in Theorem~\ref{thm:even} have one and two redundant symbols for the cases $q=2$ and $q\ge 3$, respectively. 
	To obtain the lower bound for $q=2$, we claim that every pair of distinct words $\bx$ and $\by$ in a $(n,4;\BSD_2)$-reconstruction code are not Type-B-confusable with $m=1$. 
	Suppose otherwise. Then Proposition~\ref{char:sd} implies that $|\BSD_q(\bx)\cap \BSD_q(\by)|=4$, a contradiction. Applying Theorem~\ref{thm:confusable-optimal}, we have that $\rho(n,4;\BSD_2) \ge  1 - o(1)$. When $q\ge 3$, we have that every pair of distinct words $\bx$ and $\by$ has Hamming distance at least two. Hence, $\rho(n,4;\BSD_q)\ge 1$.

	When $q\ge 4$ and $5\le N\le q+1$, from Proposition~\ref{char:sd}, we have that $\C$ is an $(n,N;\BSD_q)$-reconstruction code if and only if $\C$ is code with minimum distance two. Hence, $\rho(n,N;\BSD_q)=1$.
\end{proof}

\begin{theorem}
	Consider the error-ball $\BSI_q$. We have the following.
		
	\[
	\rho(n,N;\BSI_q) = 
	\begin{cases}
	\log_2 n +\Theta(1), & \mbox{if } N\in\{1,2\} \mbox{ and } q=2,\\
	\log_q n +O(\log_q\log_q n), & \mbox{if } N\in\{1,2\}\mbox{ and } q\ge 3,\\
	\log_q \log_q n +\Theta(1), & \mbox{if } N=3,\\
	2 - \epsilon(n)  & \mbox{if } N = 4\\
	1, & \mbox{if } 5\le N\le q+2,\\
	0, & \mbox{if } N\ge q+3.
	\end{cases}
	\]
	Here, $0\le \epsilon(n)\le \log_2 (4/3)-o(1)$ if $q=2$, and $0\le \epsilon(n)\le 1$ if $q\ge 3$.
\end{theorem}

\begin{proof}
	The case for $N\in\{1,2\}$ follows from Corollary~\ref{cor:ecc} while the case for $N\ge \max\{5,q+2\}$ follows from Proposition~\ref{char:sd}.
	
%
	When $N=3$, the codes in Theorem~\ref{thm:sd} have redundancy $\log_q\log_q n + O(1)$ and it remains to provide the corresponding lower bound.
	If $\C$ is an $(n,3;\BSI_q)$-reconstruction code, then we claim that every pair of distinct words $\bx$ and $\by$ are not Type-B-confusable. Suppose otherwise. Then Proposition~\ref{char:sd} implies that $|\BSI_q(\bx)\cap \BSI_q(\by)|\ge 3$, a contradiction. Again, applying Theorem~\ref{thm:confusable-optimal}, we have that $\rho(n,2;\BSI_q) \ge \log_q\log_q n - O(1)$.
	
	When $N=4$, the codes in Theorem~\ref{thm:even} have two redundant symbols for the cases $q=2$ and $q\ge 3$. 
	To obtain the lower bound, we claim that every pair of distinct words $\bx$ and $\by$ in a $(n,4;\BSD_2)$-reconstruction code have Hamming distance at least two and are not Type-B-confusable with $m=1$. Suppose otherwise. Then Proposition~\ref{char:sd} implies that $|\BSI_q(\bx)\cap \BSI_q(\by)|=4$, a contradiction. Hence, $\rho(n,4;\BSI_q)\ge 1$ for all $q$.
	When $q=2$, Theorem~\ref{thm:confusable-optimal} provides a stronger bound. 
	Specifically, we have that $\rho(n,4;\BSI_q) \ge  \log_2 3 - o(1)$.
	
	When $q\ge 3$ and $5\le N\le q+2$, from Proposition~\ref{char:sd}, we have that $\C$ is an $(n,N;\BSI_q)$-reconstruction code if and only if $\C$ is code with minimum distance two. Hence, $\rho(n,N;\BSI_q)=1$.
\end{proof}

\begin{theorem}
	Consider the error-ball $\Bedit_q$. We have the following.
	
	\begin{itemize}
		\item If $q=2$, we have that
		\[
		\rho(n,N;\Bedit_2) = 
		\begin{cases}
		\log_2 n +\Theta(1), & \mbox{if } N\in\{1,2\},\\
		\log_2 \log_2 n +\Theta(1), & \mbox{if } N\in\{3,4\},\\
		2-\epsilon(n), & \mbox{if } N=5,\\
		1-o(1), & \mbox{if } N=6,\\
		0, & \mbox{if } N\ge 7.
		\end{cases}
		\]
		
		Here, $0\le \epsilon(n)\le \log_2 (4/3)-o(1)$.
		
		\item If $q\ge 3$, we have that
		
		\[
		\rho(n,N;\Bedit_q) = 
		\begin{cases}
		\log_q n + O(\log_q\log_q n), & \mbox{if } N\in\{1,2\},\\
		\log_q \log_q n +\Theta(1), & \mbox{if } N\in\{3,4\},\\
		2 - \epsilon(n)  & \mbox{if } N \in \{5,6\}\\
		1, & \mbox{if } 7\le N\le q+3,\\
		0, & \mbox{if } N\ge q+4.
		\end{cases}
		\]
		Here, $0\le \epsilon(n)\le 1$. 
	\end{itemize}
\end{theorem}

\begin{proof}
	The case for $N\ge \max\{7,q+4\}$ follows from Proposition~\ref{char:edit}.
	
	When $N\in\{1,2\}$, we proceed before to show that an $(n,N;\Bedit_q)$-reconstruction code is an $(n,1;\BS_q)$-reconstruction code. Hence, $\rho(n,N;\Bedit_2)\ge \log n +\Theta(1)$. We then consider two subcases. For $q=2$, the codes in Theorem~\ref{thm:ecc} are asymptotically optimal and so, $\rho(n,N;\Bedit_2)=\log_2 n +\Theta(1)$. When $q\ge 3$, the best known $(n,N;\Bedit_q)$-reconstruction codes are those in Theorem~\ref{thm:ecc} and this gives the corresponding upper bound.
	
	When $N\in\{3,4\}$, the codes in Corollary~\ref{cor:edit} have redundancy $\log_q\log_q n + O(1)$ and it remains to provide the corresponding lower bound.
	If $\C$ is an $(n,N;\BSI_q)$-reconstruction code, then we claim that every pair of distinct words $\bx$ and $\by$ are not Type-B-confusable. Suppose otherwise. Then Proposition~\ref{char:edit} implies that $|\Bedit_q(\bx)\cap \Bedit_q(\by)|\ge 3$, a contradiction. Again, applying Theorem~\ref{thm:confusable-optimal}, we have that $\rho(n,2;\BSI_q) \ge \log_q\log_q n - O(1)$.
	
	When $q=2$ and $N=5$, or $q\ge 3$ and $N\in\{5,6\}$, the codes in Theorem~\ref{thm:even} have two redundant symbols.
	To obtain the lower bound, we proceed as before to show that every pair of distinct words $\bx$ and $\by$ in a $(n,N;\Bedit_q)$-reconstruction code have Hamming distance at least two and are not Type-B-confusable with $m=1$. 
	Hence, $\rho(n,N;\Bedit_q)\ge 1$ for all $q$.
	When $q=2$, we have that $\rho(n,4;\BSI_q) \ge  \log_2 3 - o(1)$ from Theorem~\ref{thm:confusable-optimal}.
	
	When $q=2$ and $N=6$, the codes in Theorem~\ref{thm:even} have one redundant symbol.
	For 
	the lower bound, we proceed as before to show that every pair of distinct words $\bx$ and $\by$ in a $(n,N;\Bedit_q)$-reconstruction code have Hamming distance at least two and are not Type-B-confusable with $m=1$. 
	Hence, $\rho(n,N;\Bedit_q)\ge 1$ for all $q$.
	When $q=2$, we have that $\rho(n,4;\BSI_q) \ge  \log_2 3 - o(1)$ from Theorem~\ref{thm:confusable-optimal}.

	When $q\ge 4$ and $7\le N\le q+3$, from Proposition~\ref{char:sd}, we have that $\C$ is an $(n,N;\Bedit_q)$-reconstruction code if and only if $\C$ is code with minimum distance two. Hence, $\rho(n,N;\Bedit_q)=1$.
\end{proof}

\section{Bounded Distance Decoder for Reconstruction Code}
\label{sec:simulations}

In this section, we propose a simple bounded distance decoder that attempts to reconstruct a codeword given a number of noisy reads.
Unlike previous sections, instead of assuming at most one error in each read, 
we consider probabilistic channels where each symbol is deleted, inserted or substituted with certain probabilities.
Then given a fixed number of noisy reads, we perform simulations to approximate the probability of decoder failures for the codes constructed in the previous sections.

A formal description of the bounded distance decoder is as follows.

\vspace{2mm}

\noindent{\bf Bounded Distance Decoder for Reconstruction Codes}. We are given an $(n,N;\Bedit_q)$-reconstruction code $\C$.

{\sc Input}: $Y$, a multiset of $\Nsys$ output words from $\Sigma_q^*$.\\
{\sc Output}: $\dec(Y) = \widehat{\bx} \in \C$ or $\dec(Y)= {\sf fail}$

\begin{enumerate}[(I)]
	\item For each output word $\by$, we generate a list of words that are ``closest'' to $\by$. 
	Here, we discard words whose lengths are not in $\{n-1,n,n+1\}$. Specifically,\\
	for $\by \in Y$
	
	\hspace{10mm}if $\by\in \C$, we set $L(\by)\triangleq\{\by\}$;
	
	\hspace{10mm}else if $|\by| \in \{n-1,n,n+1\}$, we set $L(\by)\triangleq\Bedit_q(\by)\cap \C$;
	
	\hspace{10mm}else if $|\by| \notin \{n-1,n,n+1\}$, we set $L(\by)\triangleq\varnothing$.
	
	Note that $L(\by)\subset \C$.
	\item We set $\widehat{\bx}$ to be the unique codeword that appears in the most number of lists.
	In other words, we set $X \triangleq \arg\max_{\bx \in \C} |\{L(\by): \by\in Y,\, \bx\in L(\by)\}|$.
	If $|X|=1$, then $\dec(Y)\triangleq\widehat{\bx}$, the unique codeword in $X$. Otherwise, $\dec(Y)\triangleq{\sf fail}$.
\end{enumerate} 

Here, the set $Y$ of $\Nsys$ output words is generated from some codeword $\bx\in\C$ and 
each symbol in $\bx$ is either deleted with probability $p_d$, 
has an insertion with probability $p_i$ or
is substituted with another symbol with probability $p_s$.
To evaluate the performance of the reconstruction codes given in this paper, 
we assume certain values of $\Nsys$ and error probabilities, and 
perform computer simulations to approximate the probability of decoding failure,
that is, $\dec(Y)={\sf fail}$ or $\widehat{\bx}\ne \bx$. 

\begin{remark}
	Note that the decoder described in this section is {\em not} a maximum likelihood (ML) decoder.
	In a parallel work~\cite{Sabary.2020}, we studied the failure probability of an ML decoder for the case of two noisy reads under a number of different coding scenarios.
	A detailed theoretical analysis of the ML decoder for the codebooks in this paper is deferred to future work.
\end{remark}

\begin{figure}[!t]
	\centering
	\begin{tabular}{ll}
		(a) $\Nsys=5$ & (b) $\Nsys=5$\\
		\includegraphics[height=6cm]{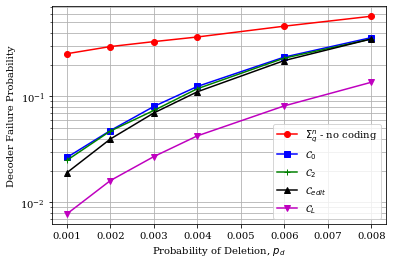} &
		\includegraphics[height=6cm]{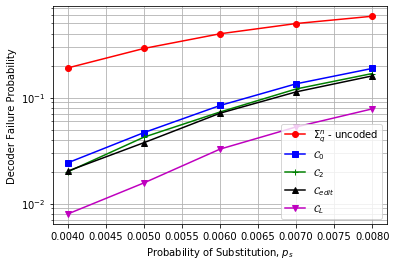}\\
		(c) $\Nsys=10$ & (d) $\Nsys=10$\\
		\includegraphics[height=6cm]{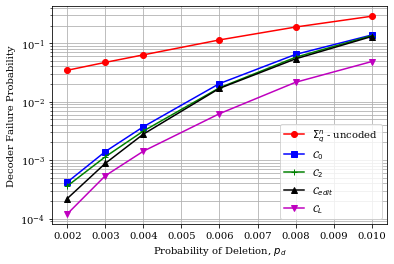} &
		\includegraphics[height=6cm]{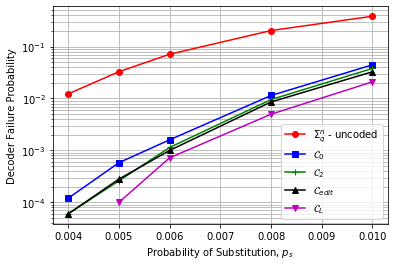}\\
		(e) $\Nsys=15$ & (f) $\Nsys=15$\\
		\includegraphics[height=6cm]{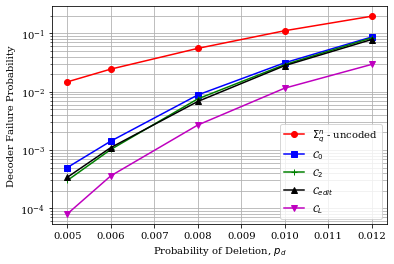} &
		\includegraphics[height=6cm]{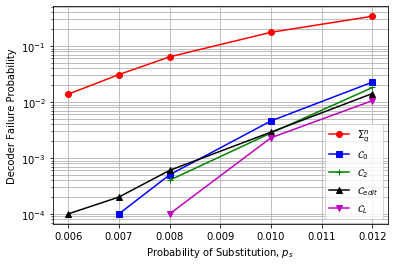}\\		
	\end{tabular}
	\caption{Decoder failure probabilities of reconstruction codes for varying read coverage and deletion probabilities. Note that $\Sigma_4^n$, $\C_0$, $\C_2$, $\Cedit$ and $\C_L$ are $(n,8;\Bedit_4)$-, $(n,7;\Bedit_4)$-, $(n,5;\Bedit_4)$-, $(n,3;\Bedit_4)$-, and $(n,1;\Bedit_4)$-reconstruction codes, respectively. Furthermore, the number of redundant symbols used is zero, one, two, $\approx 3$ and $\approx 8.248$, respectively.
	\vspace{10mm}
}
	\label{fig:simulations}
\end{figure}

For illustrative purposes, we consider the application of {\em DNA-based data storage} as described in Section~\ref{sec:intro}.
Hence, we set $q=4$, the number of nucleobases, 
and set the codeword length $n=152$, according to the capabilities of the current synthesis technologies. 
For the error probabilities, we consider a recent experiment conducted by Organick \etal \cite{Organick.2018} who stored 200MB of data in over 13 million DNA
strands and reported deletion, insertion, and substitution rates
to be $1.5 \times  10^{-3}$, $5.4 \times 10^{-4}$, and $4.5 \times 10^{-3}$, respectively.
In the same paper, the authors also surveyed previous experiments and provided a summary on the number of strands that were sampled for sequencing purposes.
The number of reads per strand vary from five to 51 (see \cite[Figure 1(e)]{Organick.2018}).
Therefore, for our simulations, we adopt the following setup.
\vspace{1mm}

\noindent{\bf Reconstruction codes}. We consider the codes $\Sigma_4^n$, $\C_0$, $\C_2$ and $\Cedit$ described in Section~\ref{sec:construction}.
When $n=152$, these codes use zero, one, two and $\approx 3$ redundant symbols, respectively. For $\Cedit$, we set the parameter $P=15$ for the construction in Corollary~\ref{cor:edit}. 
Note that $\Sigma_4^n$, $\C_0$, $\C_2$ and $\Cedit$ are $(n,8;\Bedit_4)$-, $(n,7;\Bedit_4)$-, $(n,5;\Bedit_4)$- and $(n,3;\Bedit_4)$-reconstruction codes, respectively. 
In addition to these codes, we consider an $(n,1;\Bedit_4)$-reconstruction, or equivalently, a quaternary single-edit-correcting code. When $n = 153$, the best known $(n,1;\Bedit_4)$-reconstruction $\C_L$
uses approximately $8.248$ redundant symbols. The quartenary code $\C_L$ is obtained via a modification of a binary single-edit-correcting code due to Levenshtein~\cite{Levenshtein.1966} and we refer the reader to \cite{Cai.2019} for the details of the construction.

\noindent{\bf Number of noisy reads}. Following \cite{Organick.2018}, we set $\Nsys\in \{5,10,15\}$.

\noindent{\bf Error probabilities}. Following \cite{Organick.2018}, we fix $p_i=6\times 10^{-3}$ and let
let $p_s$ and $p_d$ take values in the ranges $[5\times 10^{-3}, 1.2\times 10^{-2}]$ and $[2\times 10^{-3}, 10^{-2}]$, respectively.
\vspace{1mm}

For each reconstruction code, we pick 10,000 to 50,000 codewords uniformly at random and for each codeword $\bx$, we generate the set $Y$ according to the given error probabilities. After which, we apply the bounded distance decoder and verify whether the decoded word corresponds to $\bx$. The results are presented in Figure~\ref{fig:simulations}.
As expected, the codes that use more redundant symbols perform better in terms of the decoder failure rate. We observe that all codes (with at least one redundant symbol) outperform the uncoded case ($\Sigma_4^n$) with a reduction in decoder failure rate of one or two orders of magnitude. 
Of importance, for large numbers of reads $\Nsys$ and small deletion / substitution rates, we also observe that the coding gain of $\C_L$ (which uses $\approx 8.248$ redundant symbols) is not as significant. This suggests that for these system parameters, reconstruction codes such as $\C_0$ suffice to provide similar performance as classical error-correcting codes. In certain cases, the gain in information rate is significant. In our example where $n=152$, using $\C_0$ in lieu of $\C_L$ results in a gain in information rate of 4.6 percent.

\begin{remark}
	For the DNA-based coding scenario, we envision that the reconstruction codes proposed in this paper will be used as inner codes in a concatenated coding scheme.
	Concatenated coding schemes are commonly used in experiments (see for example, \cite{Organick.2018, Ehrlich.2017}) 
	as the current synthesis technologies are unable to sequence long codewords.
	However, to integrate with these coding schemes, one needs to address issues such as the unordered nature \cite{Lenz.2020} and biochemical constraints of \cite{Immink.2020} DNA strands. 
	Furthermore, even though the decoder proposed in this paper is a hard decision decoder, 
	modifications can be made so that a soft decision to passed to the outer code \cite{Zhang.2018}. 
	However, such issues are beyond the scope of this paper and are deferred to future work.
\end{remark}

\section{Conclusion}

We studied the sequence reconstruction problem in the context when the number of noisy reads $N$ is fixed.
Specifically, for a variety of error-balls $B$, we designed $(n,N;B)$-reconstruction codes for $2\le N \le \nu(\C;B)$ and 
derived their corresponding lower bounds.
Of significance, our code constructions use $o(\log n)$ bits of redundancy and 
in most cases are asymptotically optimal. 
Table~\ref{code:summary} summarizes all code constructions given in this paper.

The framework studied in this work is a natural extension of the classical error-correcting model where only one noisy read is available. 
Within this framework, there are several interesting research directions and 
some of which are partially addressed in our concurrent work \cite{Chrisnata.arxiv.2020, Sabary.2020}.
One direction is to extend our analysis to {the} case of more than one {error}, for example, more than one deletion.
Another direction is to anaylse the decoder failure probability of {an} ML decoder for the reconstruction codes in the paper.
A final direction is to adapt these {reconstruction algorithms} with other design considerations in specific coding scenarios.

%
%

\input{biblio.bbl}

\end{document}